\documentclass[journal]{IEEEtran}
\usepackage{enumerate}
\usepackage{cite}
\usepackage{graphicx}
\usepackage[cmex10]{amsmath}
\usepackage{textcomp,amssymb}
\usepackage{makeidx,epsfig}
\usepackage[dvipsnames]{xcolor}
\usepackage[utf8]{inputenc}
\usepackage{array}
\usepackage{amsthm}
\allowdisplaybreaks
\usepackage{algpseudocode}
\usepackage{caption}
\usepackage{diagbox} 
\usepackage{cancel}
\usepackage{multirow}

\usepackage[linesnumbered,ruled]{algorithm2e}
\SetKwInOut{Initialization}{Initializations}
\SetKwInOut{Variable}{Variables}
\SetKw{Break}{break}

\usepackage{csquotes}

\newtheorem{Pro}{\sc Proposition}
\newtheorem{Def}{\sc Definition}

\definecolor{darkgreen}{rgb}{0.0, 0.5, 0.0}

\def\ba{{\bf a}}
\def\bb{{\bf b}}
\def\bu{{\bf u}}
\def\bv{{\bf v}}

\def\bx{{\bf x}}
\def\by{{\bf y}}
\def\bz{{\bf z}}
\def\bW{{\bf W}}
\def\bX{{\bf X}}

\def\boldx{\boldsymbol{x}}
\def\tbz{\tilde{\bz}}
\def\bPsi{\boldsymbol{\Psi}}
\def\bTheta{\boldsymbol{\Theta}}
\def\bPhi{\boldsymbol{\Phi}}

\def\blambda{\boldsymbol{\lambda}}
\def\bmu{\boldsymbol{\mu}}
\def\bnu{\boldsymbol{\nu}}
\def\bxi{\boldsymbol{\xi}}

\def\bgamma{\boldsymbol{\gamma}}
\def\bGamma{\boldsymbol{\Gamma}}

\def\calC{{\cal C}}
\def\calD{{\cal D}}

\def\calN{{\cal N}}
\def\calQ{{\cal Q}}

\def\calO{{\cal O}}

\psfull
\begin{document}
\title{Compressed Sensor Caching and Collaborative Sparse Data Recovery with Anchor Alignment}
\author{Yi-Jen Yang, Ming-Hsun Yang, Jwo-Yuh Wu, and Y.-W. Peter Hong\thanks{This work was supported in part by the National Science and Technology Council (NSTC) of Taiwan under grants NSTC 111-2221-E-007-042-MY3, 111-2221-E-A49-067-MY3, 111-2222-E-008-003, and 112-2221-E-008-057. A preliminary version of this work was published in ICASSP 2019.}
\thanks{Y.-J.~Yang and Y.-W.~P.~Hong are with the Institute of Communications Engineering, National Tsing Hua University, Hsinchu 30013, Taiwan (e-mail: yjyang@erdos.ee.nthu.edu.tw; ywhong@ee.nthu.edu.tw).}
\thanks{M.-H.~Yang is with the Department of Communication Engineering, National Central University, Taoyuan 320, Taiwan (e-mail: mhyang@cc.ncu.edu.tw).}
\thanks{J.-Y.~Wu is with the Department of Electrical and Computer Engineering, the Institute of Communications Engineering, and the College of Electrical Engineering, National Yang Ming Chiao Tung University, Hsinchu 30010, Taiwan (e-mail: jywu@cc.nctu.edu.tw).}}

\maketitle

\begin{abstract}
This work examines the compressed sensor caching problem in wireless sensor networks and devises efficient distributed sparse data recovery algorithms to enable collaboration among multiple caches. In this problem, each cache is only allowed to access measurements from a small subset of sensors within its vicinity to reduce both cache size and data acquisition overhead. To enable reliable data recovery with limited access to measurements, we propose a distributed sparse data recovery method, called the collaborative sparse recovery by anchor alignment (CoSR-AA) algorithm, where collaboration among caches is enabled by aligning their locally recovered data at a few anchor nodes. The proposed algorithm is based on the consensus alternating direction method of multipliers (ADMM) algorithm but with message exchange that is reduced by considering the proposed anchor alignment strategy. Then, by the deep unfolding of the ADMM iterations, we further propose the Deep CoSR-AA algorithm that can be used to significantly reduce the number of iterations. We obtain a graph neural network architecture where message exchange is done more efficiently by an embedded autoencoder. Simulations are provided to demonstrate the effectiveness of the proposed collaborative recovery algorithms in terms of the improved reconstruction quality and the reduced communication overhead due to anchor alignment.
\end{abstract}


\section{Introduction}\label{sec.intro}


Wireless sensor networks (WSNs) and the internet of things consist of a large number of sensors that are deployed to gather information from the environment \cite{yick_etal_2008}. The information collected must often be forwarded back to the data-gathering node or made easily accessible by the users. Due to the large-scale and wide deployment of sensors, constant query of information directly from the sensors may be inefficient and costly. Therefore, it is necessary to deploy dedicated caches within the network or to better utilize on-board sensor caches to allow easy access to historical sensor data and thus reduce the frequent activation of sensors. However, the large-scale deployment of sensors and their continuous measurement of the environment generate massive data that must be stored efficiently. In this work, we explore the high compressibility of sensor data and the collaborative nature of sensor nodes to improve the effectiveness of sensor caching and data recovery.

Sensor caching problems have been examined extensively in the literature, but mostly from the networking perspective, e.g., in \cite{ying_etal_2008,dimokas_etal_2011,niyato_etal_2016,Yao2018joint, jaber_kacimi_2020, yang_song_2021, zhang_ren_wang_lv_2022, chen_sun_yang_taleb_2022}. 
Specifically,
\cite{ying_etal_2008} considered the optimal placement of caches (or storages) to minimize the total cost of storage, computation, and replying requests. 
\cite{dimokas_etal_2011} proposed the use of cooperative caching among sensors to improve the accessibility of sensor data.
\cite{niyato_etal_2016} utilized caching to improve the efficiency of energy usage in energy harvesting WSNs. Furthermore, \cite{Yao2018joint} examined the joint storage allocation and content placement problem, and proposed a hierarchical cache-enabled C-RAN to handle the huge data traffic generated by IoT sensing services. 
More recently, \cite{jaber_kacimi_2020} proposed a sensor caching strategy for content-centric networking where sensors chosen to cache the data content were determined by taking into account the node centrality and their distances from the content source.
\cite{yang_song_2021} proposed a multihop cooperative caching strategy to find the optimal tradeoff between energy saving and delay reduction by coordinating the in-network caching and the sensors' sleep schedule. \cite{zhang_ren_wang_lv_2022} proposed a joint fractional caching and multicast beamforming scheme to maximize the energy efficiency of multicast transmission with the help of cooperative transmission by in-network caches.
\cite{chen_sun_yang_taleb_2022} considered a joint caching and computing service placement problem and used deep reinforcement learning to adapt to heterogeneous systems with limited prior knowledge. While sensor caching problems are not new, most existing works
examine the problem from a networking perspective, treating sensor data as individual packets that must be delivered separately. 
Unlike these works, we address the problem from a signal processing perspective, taking advantage of the compressibility of sensor measurements to reduce the cost of caching.



Due to the locality and smoothness of many physical phenomena, the measurements collected by densely deployed sensors in a WSN are often highly correlated in both space and time \cite{vuran2004spatio}. These sensor measurements tend to be sparse when viewed under a specific sparsifying basis. Hence, compressed sensing (CS) techniques \cite{candes_wakin_michael_2008} have often been adopted to reduce the cost of data acquisition and recovery in WSNs. In particular, by exploiting correlation in the spatial domain, \cite{haupt08} proposed a general framework for using CS to efficiently compress and deliver networked data. Random projections of local data were transmitted using gossip-type algorithms in a general multihop network.
\cite{luo2009compressive} proposed a compressive data gathering scheme where weighted sums of sensor readings are sent to a common sink node. The advantages in terms of lifetime extension, robustness to abnormal readings, and improved network capacity were demonstrated using real sensor data.
\cite{lee2009spatially} proposed a spatially-localized projection technique based on the clustering of neighboring sensors. 
The impact of spatial overlap between clusters and the choice of sparsifying bases were considered in the cluster design. \cite{li_xu_wang_2013} proposed a CS-based framework for data sampling and acquisition in the internet of things. An efficient cluster-sparse reconstruction algorithm was proposed for in-network compression to improve data reconstruction and energy efficiency.

Moreover, by considering correlation in both spatial and temporal domains,
\cite{leinonen2015sequential} proposed a sequential CS recovery method 
that allows the sink node to efficiently reconstruct the correlated sensor data streams from periodically delivered CS measurements.
The proposed method utilizes estimates from previous time instants to improve the recovery at the current time instant.
\cite{chen_ranieri_zhang_vetterli_2015} also exploited the spatial and temporal characteristics of the underlying signal
to reduce the amount
of collected samples. 
The signal model was learned adaptively from the measurements and used to determine the spatial and temporal sampling of the physical field. 
\cite{Quer2012} combined the use of CS and principal component analysis (PCA) to enable online recovery of large data sets in WSNs. The proposed method uses PCA to capture the essential components in space and time and
self-adapts to changes in the signal
statistics. \cite{li_tao_chen_2018} proposed a CS-based data gathering algorithm that utilizes random sampling and random walks to perform efficient sensor data selection in both temporal and spatial domains.
However, these works do not take advantage of CS for sensor caching or perform collaborative data recovery.



In this work, we examine a compressed sensor caching problem in WSNs and devise distributed sparse data recovery algorithms to enable collaboration among caches.
Here, multiple caches are deployed in a WSN to store sensor measurements within their respective coverage areas over a certain time window. To reduce cache size and data acquisition overhead, we allow each cache to only store measurements from a small subset of sensors within its vicinity.
By exploiting the sparse nature
of the spatially-temporally correlated sensor field and the collaboration among caches, sensor data can be recovered
locally at each cache using a distributed CS-based sparse data recovery method. The main contributions of this work are summarized as follows:
\begin{itemize}
    \item We propose a compressed sensor caching strategy based on the CS framework where only a subset of sensor measurements need to be stored at the local caches to enable data recovery when requested by the users.
    \item We formulate a collaborative sparse data recovery problem that aims to reduce the recovery error (beyond that of separate recovery by each local cache) by exploiting collaboration among caches and thus taking advantage of the measurements gathered by other caches.
    
    \item We propose the collaborative sparse recovery by anchor alignment (CoSR-AA) algorithm that ensures consistency of the local recovery at different caches by aligning their estimates of the data associated with a few randomly selected anchor nodes. The proposed algorithm is based on the consensus alternating direction method of multipliers (ADMM) algorithm \cite{boyd2011distributed}, but utilizes the proposed anchor alignment to significantly reduce the number of message exchanges between caches. 
    \item Moreover, based on the deep unfolding \cite{hersey_roux_weinger_2014} of the ADMM iterations in the original CoSR-AA algorithm, we also propose the Deep CoSR-AA algorithm to further reduce the communication overhead. We obtain a graph neural network architecture where the message exchange is done more efficiently by an embedded autoencoder.

    \item Finally, simulation results are provided to demonstrate the effectiveness of the proposed CoSR-AA and Deep CoSR-AA algorithms versus other baseline data recovery methods in terms of improved reconstruction accuracy and reduced communication overhead.
\end{itemize}
A preliminary version of this work can be found in \cite{yang_yang_hong_wu_ICASSP2019}, where a solution was provided for only the two-cache case. Different from \cite{yang_yang_hong_wu_ICASSP2019}, this work considers the general multi-cache scenario and justifies the use of anchor alignment through theoretical sparsity arguments. Furthermore, we also propose a deep learning implementation to further reduce the message exchange overhead. This work is also related to our recent publication in \cite{chennakesavula_hong_scaglione_TSP2022}, where the focus is on optimizing the linear compression at the caches to enhance distributed parameter estimation (rather than sparse data recovery).

Distributed sparse recovery methods have also been investigated before in the CS literature, utilizing techniques such as iterative hard thresholding \cite{patterson_eldar_keidar_2014}, the quadratic penalty method \cite{borhani_watt_katsaggelos_2017} and the previously mentioned ADMM algorithm. The ADMM algorithm is favorable due to its ability to decompose the problem into multiple subproblems that can be solved in parallel by distributed local agents.
In particular, \cite{Mota2012distributed} proposed a distributed implementation of the ADMM algorithm to solve the optimization problem associated with the conventional basis pursuit. 
\cite{ling_tian_2010} developed a decentralized sparse signal recovery algorithm based on the consensus ADMM algorithm and utilized it to enable a random node sleeping strategy. Each active sensor only monitors and recovers its local region, but collaborates with its neighboring active sensors to iteratively improve its local estimates.
\cite{tian_zhang_hanzo_2023} considered a multi-view CS problem, where each sensor can only obtain a partial view of the global sparse vector. 
The recovery problem was formulated as a bilinear optimization problem based on a factored joint sparsity model and was again solved by an ADMM-based in-network algorithm. 
Different from the above works, where the number of message exchanges typically scales with the ambient dimension of the signal, our proposed method only exchanges data corresponding to anchor nodes, and thus significantly reduces the communication overhead.

Several recent works, e.g., \cite{mousavi_etal_2015,mousavi_etal_2017,palangi_ward_deng_2016,wu_rosca_lillicrap_2019}, have also combined the use of deep learning with CS to improve the encoding and decoding efficiency. In particular, \cite{mousavi_etal_2015} was one of the first works to adopt deep learning for the encoding and decoding of sparse signals. A stacked denoising autoencoder was used to capture statistical dependencies of the signal and improve signal recovery
performance over traditional CS approaches. Similar ideas were examined in \cite{mousavi_etal_2017} using convolutional neural networks.
Moreover, \cite{palangi_ward_deng_2016} considered the multiple measurement vector CS problem and utilized
long short-term memory 
to capture the dependency among vectors (i.e.,
the conditional probability of their non-zero entries). 
\cite{wu_rosca_lillicrap_2019} considered the use of generative models for CS to reduce the sparsity requirement and showed how the reconstruction speed and accuracy can be improved by jointly training a generator and
the reconstruction process through meta-learning.
In addition, several other works have also utilized the deep unfolding of existing iterative optimization algorithms, such as ADMM, to obtain efficient and interpretable neural network models. 
\cite{yang_sun_li_xu_2017} examined the use of CS for fast Magnetic Resonance
Imaging (MRI) and proposed a deep learning architecture to optimize the CS-MRI models based on the unfolding of the ADMM algorithm. \cite{yang_sun_li_xu_2020} 
considered a generalized CS model for image reconstruction and proposed two efficient solvers based on the deep unfolding of the ADMM algorithm for optimizing the model. The effectiveness of the proposed approaches was demonstrated for both MRI and natural images.
Both \cite{yang_sun_li_xu_2017} and \cite{yang_sun_li_xu_2020} enabled the training of ADMM penalty parameters, sparsifying transforms and shrinkage function.
\cite{kouni_paraskevopoulos_rauhut_alexandropoulos_2022} proposed a neural network architecture that unrolled the ADMM iterations for the analysis compressed sensing problem. The proposed model was used to jointly learn a redundant analysis operator for sparsification and reconstruction of the signal support. \cite{fu_monardo_huang_liu_2021} unfolded an iterative shrinkage thresholding algorithm for block-sparse signal recovery by exploiting the block-sparse structure of signals. 
The weighting matrices shared across different layers and the soft-thresholding parameters were both treated as training parameters. 
While the above works successfully demonstrate the effectiveness of using deep learning in sparse data recovery, their proposed methods are centralized and thus may result in large communication overhead when implemented in a distributed setting.
The remainder of this paper is organized as follows. In Section \ref{sec.system_model}, we describe the system model and the proposed compressed sensor caching problem. In Section \ref{sec.proposed_method}, we propose a collaborative sparse data recovery method based on the ADMM algorithm and show how anchor alignment can be used to reduce the communication cost of distributed implementation. 
Then, in Section \ref{sec.deep_unfolding}, we propose an efficient neural network implementation of the collaborative sparse recovery algorithm based on the deep unfolding of the above ADMM iterations. 
Finally, we provide the simulation results in Section \ref{sec.simulations} and conclude in Section \ref{sec.conclusion}. 


\section{System Model} \label{sec.system_model}

\begin{figure}[t]
\centering{\includegraphics[width=1.0\linewidth]{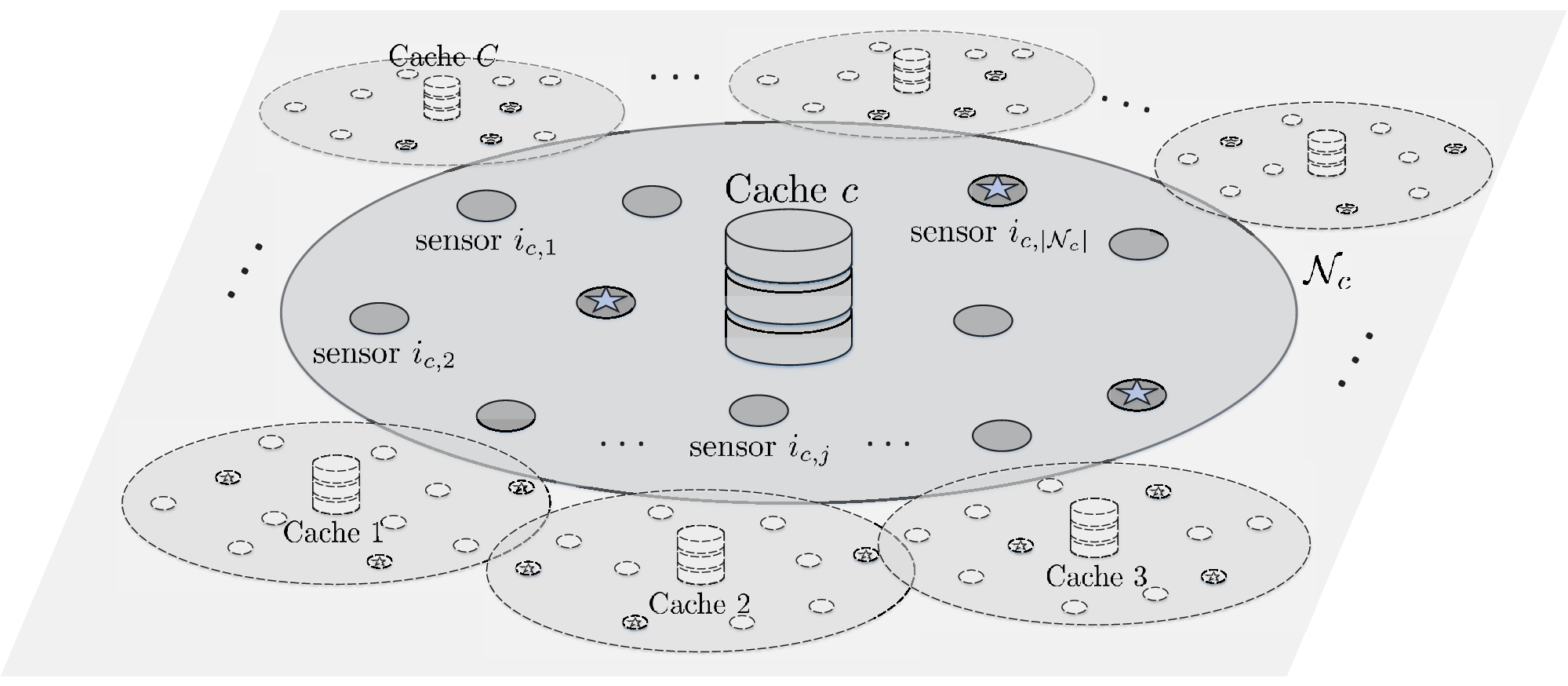}}
\caption{Illustration of the sensor caching scenario.}
\label{fig:model}
\vspace{-.2cm}
\end{figure}

We consider a WSN that consists of $N$ sensor nodes, indexed by the set $\calN=\{1,\ldots, N\}$, that are deployed to monitor the physical environment, as shown in Fig. \ref{fig:model}. The network also consists of $C$ caches, indexed by the set $\calC=\{1,\ldots, C\}$, that are placed to buffer the most recent observations made by the sensors within their vicinity.
The set of sensors accessible by cache $c$ is indexed by $\calN_c\triangleq\{i_{c,1},\ldots, i_{c,|\calN_c|}\}$, where $|\calN_c|$ is the cardinality of the set and $\cup_{c=1}^C\calN_c=\calN$. Note that $\calN_1,\ldots, \calN_C$ are not necessarily disjoint. The above setup allows remote users to access data from the caches instead of directly from the sensors to avoid frequent sensor activation.

Let $x_{n}(t)$, $n = 1, \ldots, N$, be the scalar observation made by sensor $n$ at time $t$. By collecting the observations that are made by all $N$ sensors over the recent $W$ time instants, we obtain the
$N \times W$ matrix
\begin{align}
	\bX(t) &\triangleq
	\begin{bmatrix}
		x_{1}(t-W+1) 	& \cdots	& x_{1}(t)  \\
		\vdots 		& \ddots	& \vdots    \\
		x_{N}(t-W+1)	& \cdots	& x_{N}(t)
	\end{bmatrix}\\
    &= \begin{bmatrix}
    	\underline{\boldx}(t-W+1), \ldots, \underline{\boldx}(t)
    \end{bmatrix}\\
    &= \begin{bmatrix}
    	\boldx_1(t), \ldots, \boldx_{N}(t)
    \end{bmatrix}^T,
\end{align}
where 
$\underline{\boldx}(t) \triangleq \begin{bmatrix}
    x_{1}(t), \ldots, x_{N}(t)
\end{bmatrix}^T$ is the vector of all sensor observations gathered at time $t$, and $\boldx_{j}(t) \triangleq \begin{bmatrix}
    x_{j}(t-W+1), \ldots, x_{j}(t)
\end{bmatrix}^T$ is the vector of observations made by sensor $j\in\calN$ over time instants $t-W+1$ to $t$.

Since the observed data are spatially correlated across sensors and temporally correlated over consecutive time instants, we assume the existence of
sparsifying bases $\bPsi_{S} \in \mathbb{R}^{N \times N}$ and $\bPsi_{T} \in \mathbb{R}^{W\times W}$ in the spatial and temporal domains, respectively, such that $\underline{\boldx}(t) = \bPsi_{S} \boldsymbol{\theta}_{S}(t)$ and $\boldx_{n}(t) = \bPsi_{T} \boldsymbol{\theta}_{n,T}(t)$, for all $t$ and $n$, where $\boldsymbol \theta_{S}(t) \in \mathbb{R}^{N}$ and $\boldsymbol \theta_{n,T}(t) \in \mathbb{R}^{W}$ are sparse vectors.
Then, by adopting the Kronecker compressed sensing (KCS) framework \cite{duarte2012kronecker} and by following the derivations in \cite{leinonen2015sequential}, 
we obtain the joint spatial-temporal sparsity model described as follows:
\begin{align}
\bx(t)={\rm vec}\left(\bX(t)\right)
&= {\rm vec}\left(\bPsi_{S}\bTheta_{S}(t)\right) \\
 &= {\rm vec}\left(\bPsi_{S} \bTheta_{S}(t) \bPsi_{T}^{-T} \bPsi_{T}^T\right) \\
 &= {\rm vec}\left(\bPsi_{S} {\bf Z}(t) \bPsi_{T}^T\right) \label{eq.2Ddct} \\
 &= \left(\bPsi_{T} \otimes \bPsi_{S}\right){\rm vec}\left({\bf Z}(t)\right)\label{eq.kronecker_basis} \\
 &=\bPsi\bz(t),
\end{align}
where ${\bf Z}(t) = \bTheta_{S}(t) \bPsi_{T}^{-T} \in \mathbb{R}^{N\times W}$ with ${\bTheta_{S}}(t) \triangleq \begin{bmatrix}
\boldsymbol{\theta}_{S}(t-W+1), \ldots, \boldsymbol{\theta}_{S}(t)
\end{bmatrix}$, $\bz(t) = {\rm vec}({\bf Z}(t))$, and $\bPsi \triangleq \bPsi_{T} \otimes \bPsi_{S} \in \mathbb{R}^{ NW \times NW}$ is a  Kronecker sparsifying basis \cite{duarte2012kronecker}.

We assume that each cache, say cache $c$, only samples observations from $M_c(\ll |\calN_c|)$ randomly selected sensors within the set $\calN_c$
at each time instant. The set of sampled sensors at time $t$ is denoted by $\mathcal{M}_c(t)$, where $|\mathcal{M}_c(t)|=M_c$, and the collected measurements can be written as
\begin{equation}
\underline{\by}_c(t) =\begin{bmatrix}
    y_{c,1}(t), \ldots, y_{c,M_c}(t)
\end{bmatrix}^T= \underline{\bPhi}_c(t) \underline{\bx}(t),
 \end{equation}
where $\underline{\bPhi}_c(t) \in \mathbb{R}^{{M_c} \times N}$ is a sensor selection matrix with each row being a canonical vector with only one entry equal to $1$ and all others equal to $0$. The position of the $1$ in each row of $\bPhi_c(t)$ indicates the sensor that is sampled.
The measurements gathered by cache $c$ over the most recent $W$ time instants
are stored in the cache at time $t$ and can be stacked into the vector
\begin{align}
\by_{c}(t)&\triangleq
\begin{bmatrix}
    \underline{\by}_c^T(t-W+1),
    \ldots, \underline{\by}_c^T(t)
\end{bmatrix}^T\\
&= \!\begin{bmatrix}
	\underline{\bPhi}_c(t\!-\!W\!+\!1) &\!\cdots &{\bf 0} \\
    \vdots 		& \!\ddots	& \vdots    \\
    {\bf 0}		& \!\cdots	& {\bf \underline{\Phi}}_c(t)
\end{bmatrix}
\begin{bmatrix}
	\underline{\bx}(t\!-\!W\!+\!1) \\
    \vdots \\
    \underline{\bx}(t)
\end{bmatrix}\\
&= \bPhi_c(t) \bx(t),
\end{align}
where
$\bPhi_c(t) = {\rm diag}(\underline{\bPhi}_c(t-W+1), \cdots,  \underline{\bPhi}_c(t))$ is an ${M_cW} \times NW$ measurement matrix.
In the centralized scenario, a dedicated fusion center (FC) can collect all cache measurements $\by_c(t)$ and perform sensor data reconstruction by solving a classical $\ell_1$-minimization problem.
The centralized sparse data recovery problem can be formulated as
\begin{subequations}\label{eq.centralized_reconstruct}
\begin{align}
\min_{\bz(t)}&~\left\|\bz(t)\right\|_1\\
\text{subject to }&~ \begin{bmatrix}
    \by_1(t)\\ \vdots\\ \by_C(t)
\end{bmatrix} = \begin{bmatrix}
    \bPhi_1(t) \\
    \vdots \\
    \bPhi_C(t)
\end{bmatrix}\bPsi\bz(t). \label{eq.centralized_constraint}
\end{align}
\end{subequations}
By recoverying the underlying sparse vector $\bz(t)$, we are able to reconstruct the entire sensor field $\bx(t)$ by multiplying $\bz(t)$ with the Kronecker sparsifying basis 
$\bPsi$.  

In the compressed sensor caching problem under consideration, each cache may only have access to a limited amount of data collected from sensors within its vicinity, namely, $\by_c(t)$ for cache $c$. In this case, local data recovery at cache $c$ based on $\by_c(t)$ may be unreliable. In the following section, we develop a distributed sparse data recovery algorithm that improves local recovery by enabling collaboration among caches.


\section{Collaborative Sparse Data Recovery with Anchor Alignment} \label{sec.proposed_method}

\begin{figure}
\centering{\includegraphics[width=1.0\linewidth]{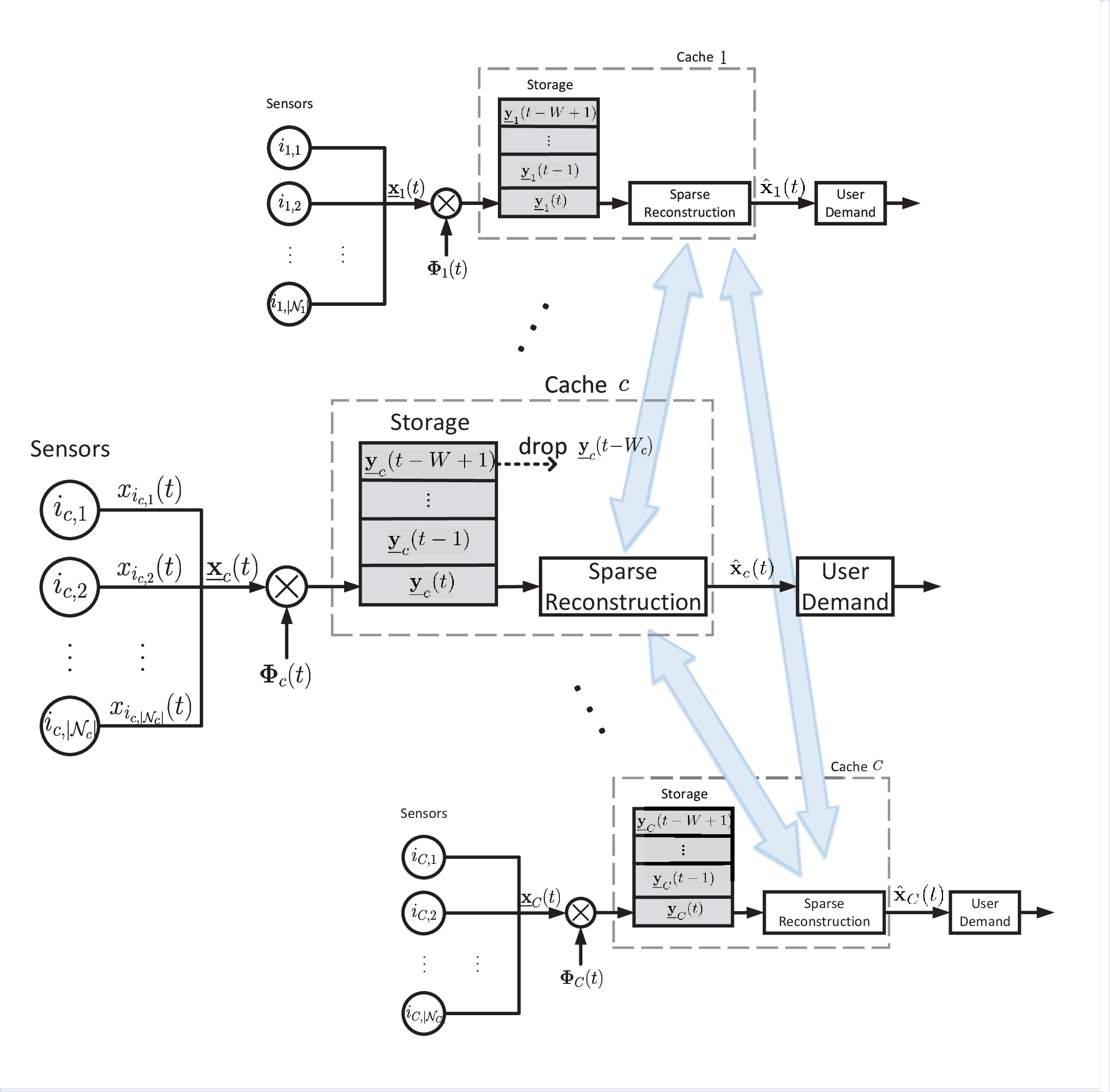}}
\caption{Model of the compressed sensor caching and data recovery operations at the multiple caches.} \label{fig:collaborative_reconstruction}
\vspace{-.2cm}
\end{figure}

In this section, we formulate and solve a distributed sparse data recovery problem that enables collaboration among local sensor caches, as illustrated in Fig. \ref{fig:collaborative_reconstruction}. 
Due to the limited storage capacity of caches and the cost of sensor communication, we allow each cache to collect and store only the observations obtained from certain sensors within its vicinity. In this case, each cache would not be able to obtain a reliable recovery of the entire sensor field and, thus, requires collaboration with other caches. In the following, we develop the proposed collaborative sparse data recovery algorithm and show how aligning the recovered data at only a limited number of anchor nodes is sufficient to exploit the advantages of collaboration.



\subsection{Collaborative Sparse Data Recovery Problem} \label{sec.multiple_cache}

To enable distributed implementation of the sparse data recovery problem in \eqref{eq.centralized_reconstruct}, we devise an approach to optimally decompose the problem into several subproblems that can be solved in parallel at the local caches. Specifically, by letting $\bz_c(t)$ be the sparse vector to be reconstructed at cache $c$, we can reformulate the sparse data recovery problem in \eqref{eq.centralized_reconstruct} as
\begin{subequations}\label{eq.distributed_reconstruct}
\begin{align}
\min_{\left\{\bz_{c}(t)\right\}_{c=1}^C}&~ \frac{1}{C}\sum_{c=1}^C\left\|\bz_{c}(t)\right\|_1\\
\text{subject to }&~ {\by}_{c}(t) \!=\! \bPhi_c(t) \bPsi\bz_{c}(t) , ~\forall c\in\calC \label{eq.observation_constraint}\\
&~\bz_{c}(t) \!=\! \bz_{c'}(t), ~\forall c'\in\calD_c, c\in\calC, \label{eq.consensus_constraint}
\end{align}
\end{subequations}
where $\calD_c\triangleq \{i_{c,1}, \ldots,i_{c,|\calD_c|}\}\subset \calC$ is the set of neighboring caches of cache $c$. 
We assume that all caches are connected, that is, there exists a path between any two caches within the network.
The optimization problems in \eqref{eq.centralized_reconstruct} and \eqref{eq.distributed_reconstruct} are equivalent since \eqref{eq.observation_constraint} is obtained by writing \eqref{eq.centralized_constraint} separately for each cache and \eqref{eq.consensus_constraint} ensures that all caches can reach a consensus on their recovery, i.e., $\bz_c(t)=\bz(t)$, for all $c$.

The optimization problem in \eqref{eq.distributed_reconstruct} enables a distributed implementation where
each cache strives to minimize $\|\bz_c(t)\|_1$ from its own measurements $\by_c(t)=\bPhi_c(t)\bPsi\bz_c(t)$, along with a certain in-built consensus protocol among caches to ensure \eqref{eq.consensus_constraint}. However, since caches can only collect observations from sensor nodes inside their covered regions, recovery of data beyond reach is typically subject to large errors. Meanwhile, consensus on all parameters in $\bz_c(t)$, as required in
\eqref{eq.consensus_constraint}, would entail a large communication cost dedicated to exchanging the estimated values of $\bz_c(t)$, which has a high spatial-temporal dimension. To resolve this drawback, we propose the use of ``anchor nodes'', on which the reconstructed sensor observations are imposed to be identical (i.e., ``aligned'') across all neighboring caches. In principle, such consistency on anchor node observations can introduce coupling between measurements from different caches, offering each cache side information about the observations collected by other caches.
To some extent, anchor nodes serve as ``pivots" against which the reconstructed sensing fields from multiple caches are  stitched
together. Furthermore, since the number of anchor nodes is much smaller than the ambient signal dimension, 
the required communication cost can be drastically reduced. 

Specifically, let $\calQ_{c,c'}(t)\subset \calN$ be the set of anchor nodes selected by the pair of caches $c$ and $c'$ at time instant $t$. The anchor nodes can be chosen independently by different pairs of collaborating caches, but are common within each pair, i.e., $\calQ_{c,c'}(t) = \calQ_{c',c}(t)$, for all $c$ and $c'$. Our proposed anchor alignment strategy proposes to replace \eqref{eq.consensus_constraint} with the following measurement consistency condition, which aligns the recovered observations associated with the selected set of anchor nodes, that is,
\begin{align} \label{eq.anchor_node_constraint}
 \bGamma_{c,c'}(t)\bPsi\bz_{c}(t) = \bGamma_{c,c'}(t)\bPsi\bz_{c'}(t), \forall c'\in\calD_c, c\in\calC,   
\end{align}
where ${\bGamma}_{c,c'}(t)={\rm diag}(\underline{\bGamma}_{c,c'}(t-W+1),\ldots,\underline{\bGamma}_{c,c'}(t))\in\mathbb{R}^{|\calQ_{c,c'}(t)| W\times N W}$ is  block diagonal with $\underline{\bGamma}_{c,c'}(t)\in\mathbb{R}^{ |\calQ_{c,c'}(t)|\times N}$ being an anchor node selection matrix 
with each row being a canonical vector with only one entry equal to $1$ and all others equal to $0$. The position of the $1$ in each row 
indicates the index of a selected anchor node. 


In the following, we provide a sufficient condition that guarantees the equivalence between the consistency on anchor-node observations in \eqref{eq.anchor_node_constraint} and the consistency of sparse representations in \eqref{eq.consensus_constraint}. Our sufficient condition is characterized in terms of the {\it restricted isometric property (RIP)}.

\begin{Def}[Restricted Isometric Property (RIP) \cite{candes_wakin_michael_2008}]
    Let ${\bf D}\in\mathbb{R}^{M\times N}$ be a matrix. If there exists a constant $\delta_S\in(0,1)$ called restricted isometry constant such that
\begin{align} \label{eq.RIP}
    (1-\delta_{S})\|\bz\|_2^2 \leq \|{\bf D}\bz\|_2^2 \leq (1+\delta_{S})\|\bz\|_2^2
\end{align}
for all $S$-sparse signals  $\bz\in\mathbb{R}^{N}$, then ${\bf D}$ is said to satisfy the restricted isometry property of order $S$ with constant $\delta_S$.
\end{Def}

\begin{Pro}
    If \(\bGamma_{c,c'}(t)\bPsi\) satisfies RIP of order \(2S\), then \eqref{eq.anchor_node_constraint} is equivalent to \eqref{eq.consensus_constraint}.    
\end{Pro}

\begin{proof} 
To show the equivalence, we need to show that 
\eqref{eq.anchor_node_constraint} implies \eqref{eq.consensus_constraint} and vice versa. The latter follows directly by substituting \eqref{eq.consensus_constraint} into \eqref{eq.anchor_node_constraint}.
The former is less obvious but can be proved by contradiction utilizing the RIP defined above. In particular, suppose that there exist two distinct $S$-sparse vectors $\bz_c(t)$ and $\bz_{c'}(t)$, for some $c\in\calC$ and $c'\in\calD_c$, satisfying $\bGamma_{c,c'}(t)\bPsi\bz_c(t) = \bGamma_{c,c'}(t)\bPsi\bz_{c'}(t)$, or equivalently $\bGamma_{c,c'}(t)\bPsi(\bz_c(t) - \bz_{c'}(t)) = \mathbf{0}$. Notice that $\bz_c(t) - \bz_{c'}(t)$ is a $2S$-sparse signal since both $\bz_c(t)$ and $\bz_{c'}(t)$ are $S$-sparse, and that $\|\bz_c(t) - \bz_{c'}(t)\|_2^2>0$ since $\bz_c(t)$ and $\bz_{c'}(t)$ are distinct.
Then, by the RIP condition in \eqref{eq.RIP}, we know that
   $(1-\delta_{2S})\|\bz_c-\bz_{c'}\|_2^2 \leq \|\bGamma_{c,c'}\bPsi(\bz_c-\bz_{c'})\|_2^2 = 0$, which contradicts with the fact that $\|\bz_c-\bz_{c'}\|_2>0$. Hence, it follows that \eqref{eq.anchor_node_constraint} implies \eqref{eq.consensus_constraint}.
\end{proof}



Note that $\bGamma_{c,c'}(t)\bPsi$ consists of $|\calQ_{c,c'}(t)|$ rows of the spatial-temporal sparsifying matrix $\bPsi$. With a sufficient number of anchor nodes, we expect $\bGamma_{c,c'}(t)\bPsi$ to eventually satisfy RIP of order $2S$, allowing \eqref{eq.consensus_constraint} and \eqref{eq.anchor_node_constraint} to become equivalent.
In our simulation study (c.f. Section \ref{sec.simulations}), we choose $\bPsi$ to be the 2D-DCT matrix. Even though a theoretical sample complexity (here, the amount of anchor nodes) ensuring RIP is hard to obtain, such a partial 2D-DCT matrix has been proved successful in many practical applications of CS \cite{stankovic_mandic_dakovic_kisil_202, zammit_wassell_2020}.

Following Proposition 1, we replace 
\eqref{eq.consensus_constraint} with \eqref{eq.anchor_node_constraint} and equivalently reformulate the collaborative sparse data recovery problem as
\begin{subequations}\label{eq.distributed_reconstruct_with_AN}
\begin{align}
\min_{\left\{\bz_{c}(t)\right\}_{c=1}^C}&\sum_{c=1}^C\left\|\bz_{c}(t)\right\|_1\\
\text{subject to }&~ {\by}_{c}(t) = \bPhi_c(t) \bPsi\bz_{c}(t) , ~\forall c\in\calC\\
&~\bGamma_{c,c'}(t) \bPsi \bz_{c}(t) \!=\! \bGamma_{c,c'}(t) \bPsi \bz_{c'}(t),\notag\\
&~\forall c'\in\calD_c, c\in\calC. \label{eq.anchor_constraint}
\end{align}
\end{subequations}
The solution to \eqref{eq.distributed_reconstruct_with_AN} is denoted by $\hat{\bz}_{c}(t)$, for $c=1,\ldots, C$. Then, the reconstructed sensor observations at cache $c$ can be computed as $\hat{\bf x}_{c}(t)=\bPsi\hat{\bz}_{c}(t)$.

\subsection{Solution based on the Consensus ADMM Algorithm}

To enable distributed implementation of the collaborative sparse recovery problem in \eqref{eq.distributed_reconstruct_with_AN}, we adopt the consensus ADMM algorithm to decompose the problem into multiple subproblems that can be solved in parallel by local caches.
To simplify our notations, we shall omit the time index $t$ in the following and 
replace ${\bz}_{c}(t)$, $\by_c(t)$, $\bPhi_c(t)$, $\bGamma_{c,c'}(t)$ and $\calQ_{c,c'}(t)$ with ${\bz}_c$, $\by_c$, $\bPhi_c$, $\bGamma_{c,c'}$ and $\calQ_{c,c'}$, respectively. By introducing the auxiliary variables $\tbz_c$ and $\bv_{c,c'}$, for all $c$ and $c'$, the
optimization problem in \eqref{eq.distributed_reconstruct_with_AN} can be reformulated as
\begin{subequations} \label{eq.admm}
\begin{align}
\mathop{\min}_{
    \substack{\bz_c, \tbz_c, \bv_{c,c'},\\
    \forall c'\in\calD_c, c\in\calC}
 } \quad &\sum_{c=1}^{C} \left\|\tbz_c \right\|_1 \label{eq.admm_a} \\
\text{subject to } ~~& \bPhi_c \bPsi \bz_{c} = \by_c, ~\forall c\in\calC \label{eq.admm_b}\\
~~& \bGamma_{c,c'}\bPsi\bz_c = \bv_{c,c'}, ~\forall c' \in \calD_{c}, c \in \calC, \label{eq.admm_c}\\
~~& \bGamma_{c',c}\bPsi\bz_{c'} = \bv_{c,c'}, ~\forall c' \in \calD_{c}, c \in \calC, \label{eq.admm_d}\\
~~& \bz_c = \tbz_c, ~\forall c\in\calC. \label{eq.admm_e}
\end{align}
\end{subequations}
Here, $\bv_{c,c'}$ represents the consensus variable between $c$ and $c'$, which is used to ensure consistency among the recovered anchor observations at the two caches, 
and $\tbz_c$ is an auxiliary variable that enables us to decouple the nonsmooth $\ell_1$-norm term from the smooth regularization terms to be seen in the following. Notice that \eqref{eq.admm} is in a standard ADMM form and thus can be solved by the consensus ADMM algorithm \cite{boyd2011distributed}.

In particular, let $\blambda_c\in\mathbb{R}^{ M_cW}$, $ \bmu_{c,c'}\in\mathbb{R}^{|\calQ_{c,c'}|W}$,$ \bnu_{c',c}\in\mathbb{R}^{|\calQ_{c,c'}|W}$ and $\bxi_c\in\mathbb{R}^{NW}$ be the multipliers corresponding to \eqref{eq.admm_b}, \eqref{eq.admm_c}, \eqref{eq.admm_d} and \eqref{eq.admm_e}, respectively. Then, the augmented Lagrangian function can be written as
\begin{align}
    L_{\rho}&(\bz, \bv, \bgamma)\notag\\
    =&\sum_{c=1}^{C}\bigg\{
    \left\|\tbz_c\right\|_1 + \blambda_c^T\left(\bPhi_c\bPsi\bz_c-\by_c\right) + \bxi_c^T\left(\bz_c-\tbz_c\right)\notag\\
    &+\sum_{c'\in\calD_c}\Big[ \bmu_{c,c'}^T\left(\bGamma_{c,c'}\bPsi\bz_c - \bv_{c,c'}\right)\notag\\
    &+\bnu^T_{c',c}\left(\bGamma_{c,c'}\bPsi\bz_c - \bv_{c',c}\right)\Big] + \frac{\rho}{2}\left\|\bPhi\bPsi\bz_c-\by_c\right\|_2^2\notag\\ &+ \frac{\rho}{2}\left\|\bz_c-\tbz_c\right\|_2^2 + \frac{\rho}{2}\sum_{c'\in\calD_c}\Big[\left\|\bGamma_{c,c'}\bPsi\bz_c - \bv_{c,c'}\right\|_2^2 \notag\\ &+\left\|\bGamma_{c,c'}\bPsi\bz_c - \bv_{c',c}\right\|_2^2\Big]
    \bigg\}, \label{eq.lagrangian}
\end{align}
where $\bgamma = [\bgamma_1^T, \ldots, \bgamma_C^T]^T$ with $\bgamma_c = \big[
    \blambda_c^T, \bmu_{c,i_{c,1}}^T, \bnu_{i_{c,1},c}^T,\ldots,$ $ \bmu_{c,i_{c,|\calD_c|}}^T, \bnu_{i_{c,|\calD_c|},c}^T, \bxi_c^T
\big ]^T$, and $\rho>0$ is a penalty parameter.
Notice that the augmented Lagrangian in \eqref{eq.lagrangian} can be viewed as the sum of $C$ sub-Lagragian functions, 
each of which acts as the objective function of the subproblem to be solved locally at each cache in the distributed ADMM algorithm. That is, in the $k$-th iteration of the ADMM algorithm, 
the local parameters of cache $c$ are updated according to the following equations:
\begin{subequations} \allowdisplaybreaks
\label{eq.update}
\begin{align}
    &\blambda_c^{[k]} = \blambda_c^{[k-1]} + \rho\Big(\bPhi_c\bPsi\bz_c^{[k-1]}\!-\!\by_c\Big), \label{eq.update_lambda}\\
    &\bmu_{c,c'}^{[k]} = \bmu_{c,c'}^{[k-1]} + \rho\Big(\bGamma_{c,c'}\bPsi\bz_c^{[k-1]}\!-\!\bv_{c,c'}^{[k-1]}\Big),\,\forall c'\!\in\!\calD_c, \label{eq.update_mu}\\
    &\bnu_{c',c}^{[k]} = \bnu_{c',c}^{[k-1]} + \rho\Big(\bGamma_{c,c'}\bPsi\bz_c^{[k-1]}\!-\!\bv_{c',c}^{[k-1]}\Big),\,\forall c'\!\in\!\calD_c, \label{eq.update_nu}\\
    &\bxi_c^{[k]} = \bxi_c^{[k-1]} + \rho\Big(\bz_c^{[k-1]}-\tbz_c^{[k-1]}\Big), \label{eq.update_xi}\\
    &\bz_c^{[k]} = \mathop{\arg\min}_{\bz_c} \bigg\{{\blambda_c^{[k]}}^T\big(\bPhi_c\bPsi\bz_c-\by_c\big)+ {\bxi_c^{[k]}}^T\big(\bz_c-\tbz_c^{[k-1]}\big)\notag\\ &\hspace{2.2cm}+\sum_{c'\in\calD_c}\Big[{\bmu_{c,c'}^{[k]}}^{\!\!T}\big(\bGamma_{c,c'}\bPsi\bz_c-\bv_{c,c'}^{[k-1]}\big)\notag\\
    &\hspace{2.2cm}+ {\bnu_{c',c}^{[k]}}^{\!\!T}\big(\bGamma_{c,c'}\bPsi\bz_c-\bv_{c',c}^{[k-1]}\big)\Big]\notag\\
    &\hspace{2.2cm}+ \frac{\rho}{2}\big\|\bPhi_c\bPsi\bz_c-\by_c\big\|_2^2 + \frac{\rho}{2}\big\|\bz_c-\tbz_c^{[k-1]}\big\|_2^2\notag\\
    &\hspace{2.2cm}+ \frac{\rho}{2}\sum_{c'\in\calD_c}\Big[\big\|\bGamma_{c,c'}\bPsi\bz_c-\bv_{c,c'}^{[k-1]}\big\|_2^2\notag\\
    &\hspace{2.2cm}+ \big\|\bGamma_{c,c'}\bPsi\bz_c-\bv_{c',c}^{[k-1]}\big\|_2^2\Big]\bigg\}, \label{eq.update_first_step}\\
    &\Big(\tbz_c^{[k]}, \{\bv_{c,c'}^{[k]}\}_{c'\in\calD_c}\Big) \notag\\
    &~~= \mathop{\arg\min}_{\tbz_c, \{\bv_{c,c'}\}} \bigg\{ \big\|\tbz_c\big\|_1 + {\bxi_c^{[k]}}^T\big(\bz_c^{[k]}-\tbz_c\big)\notag\\
    &\hspace{2cm}+ \sum_{c'\in\calD_c}\Big[{\bmu_{c,c'}^{[k]}}^{\!\!T}\big(\bGamma_{c,c'}\bPsi\bz_c^{[k]}-\bv_{c,c'}\big)\notag\\
    &\hspace{2cm}+ {\bnu_{c,c'}^{[k]}}^{\!\!T}\big(\bGamma_{c',c}\bPsi\bz_{c'}^{[k]}-\bv_{c,c'}\big)\Big]\notag\\
    &\hspace{2cm}+ \frac{\rho}{2}\big\|\bz_c^{[k]}-\tbz_c\big\|_2^2\notag\\
    &\hspace{2cm}+ \frac{\rho}{2}\sum_{c'\in\calD_c}\Big[\big\|\bGamma_{c,c'}\bPsi\bz_c^{[k]}-\bv_{c,c'}\big\|_2^2\notag\\
    &\hspace{2cm}+ \big\|\bGamma_{c',c}\bPsi\bz_{c'}^{[k]}-\bv_{c,c'}\big\|_2^2\Big]\bigg\}
     \label{eq.update_second_step}
\end{align}
\end{subequations}
where the variables with superscript $[k]$ stand for their values in the $k$-th iteration.

Notice that \eqref{eq.update_first_step} and \eqref{eq.update_second_step} can be derived in closed form. In particular,
by taking the derivative of the objection function in \eqref{eq.update_first_step} with respect to $\bz_c$ and setting it to zero, we get
\begin{align} \label{eq.update_z}
    \bz_c^{[k]} =& \frac{1}{\rho}\bigg[\left(\bPhi_c\bPsi\right)^T\bPhi_c\bPsi + 2\sum_{c'\in\calD_c}\left(\bGamma_{c,c'}\bPsi\right)^T\bGamma_{c,c'}\bPsi+ {\bf I}\bigg]^{-1}\notag\\
    &\cdot\bigg[\left(\bPhi_c\bPsi\right)^T\Big(\rho\by_c - \blambda_c^{[k]}\Big) + \rho\tbz_c^{[k-1]} - \bxi_c^{[k]}\notag\\
    &~~~~+ \rho\sum_{c'\in\calD_c}\left(\bGamma_{c,c'}\bPsi\right)^T\Big(\bv_{c,c'}^{[k-1]} + \bv_{c',c}^{[k-1]}\Big)\notag\\
    &~~~~- \sum_{c'\in\calD_c}\left(\bGamma_{c,c'}\bPsi\right)^T\Big(\bmu_{c,c'}^{[k]} + \bnu_{c',c}^{[k]}\Big)\bigg].
\end{align}
Moreover, to solve for $\tilde \bz_c$ in \eqref{eq.update_second_step}, we optimize each entry of $\tilde \bz_c$, i.e., $\tilde z_{c,i}$, for $i=1,\ldots, NW$, by first finding their optimal values in the ranges $\tilde z_{c,i}>0$ and $\tilde z_{c,i}<0$ separately. Then, by combining the results, the solution can be summarized as follows:
\begin{align} \label{eq.update_tz}
    \tilde z_{c,i}^{[k]} = {\cal S}_{1/\rho}\bigg(z_{c,i}^{[k]}+\frac{\xi_{c,i}^{[k]}}{\rho}\bigg), ~i = 1, \cdots, NW,
\end{align}
where $z_{c,i}^{[k]}$ and $\xi_{c,i}^{[k]}$ represent the $i$-th entry of vectors $\bz_c^{[k]}$ and $\bxi_c^{[k]}$, respectively, 
and 
${\cal S}_{\kappa}(a)$ is the soft thresholding operator defined such that ${\cal S}_{\kappa}(a) = a-\kappa$ ($a+\kappa$, respectively) if $a > \kappa$ ($a < -\kappa$, respectively), and $0$, otherwise.
Similarly, by taking the derivative of the objective function of \eqref{eq.update_second_step} with respect to $\bv_{c,c'}$ and setting it to zero, we obtain
\begin{align} \label{eq.differential_v2}
    \bv_{c,c'}^{[k]} =& \frac{1}{2}\left(\bGamma_{c,c'}\bPsi{\bz}_c^{[k]} + \bGamma_{c',c}\bPsi{\bz}_{c'}^{[k]}\right) + \frac{\bmu_{c,c'}^{[k]} + \bnu_{c,c'}^{[k]}}{2\rho},
\end{align}
for all $c'\in\calD_c$.
Moreover,
by summing \eqref{eq.update_mu} and \eqref{eq.update_nu}, it follows that
\begin{align}
    \bmu_{c,c'}^{[k]} + \bnu_{c,c'}^{[k]} 
    =& \bmu_{c,c'}^{[k-1]} + \bnu_{c,c'}^{[k-1]} - 2\rho\bv_{c,c'}^{[k-1]}\notag\\
    &+ \rho\Big(\bGamma_{c,c'}\bPsi{\bz}_c^{[k-1]} + \bGamma_{c',c}\bPsi{\bz}_{c'}^{[k-1]}\Big).\label{eq.mu_add_nu}
\end{align}
By substituting \eqref{eq.differential_v2} into \eqref{eq.mu_add_nu}, we can see that $\bmu_{c,c'}^{[k]} + \bnu_{c,c'}^{[k]} = {\bf 0}$.
Thus, \eqref{eq.differential_v2} becomes
\begin{align} \label{eq.update_v}
    \bv_{c,c'}^{[k]} =& \frac{1}{2}\left(\bGamma_{c,c'}\bPsi{\bz}_c^{[k]} + \bGamma_{c',c}\bPsi{\bz}_{c'}^{[k]}\right).
\end{align}
Notably, $\bv^{[k]}_{c,c'} = \bv^{[k]}_{c',c}$ due to symmetry.

To perform the above calculations, it is necessary for cache $c$ to receive updates of the reconstructed anchor node observations from all caches within its vicinity, i.e., $\bGamma_{c',c}\bPsi\bz_{c'}^{[k]}$, for all $c'\in\calD_c$.
The consensus variable $\bv^{[k]}_{c,c'}$ can then be updated by taking the mean of the reconstructed anchor node observations at caches $c$ and $c'$, as in \eqref{eq.update_v}. Moreover, by using \eqref{eq.update_v}, the update equations of multipliers $\bmu_{c,c'}^{[k]}$ and $\bnu_{c',c}^{[k]}$ in \eqref{eq.update_mu} and \eqref{eq.update_nu}, respectively,
can be expressed as
\begin{subequations} \label{eq.update_mu_nu}
\begin{align}
    \bmu_{c,c'}^{[k]} =& \,\bmu_{c,c'}^{[k-1]} + \frac{\rho}{2}\left(\bGamma_{c,c'}\bPsi\bz_c^{[k-1]}-\bGamma_{c',c}\bPsi\bz_{c'}^{[k-1]}\right), \label{eq.update3_mu}\\
    \bnu_{c',c}^{[k]} =& \,\bnu_{c',c}^{[k-1]} + \frac{\rho}{2}\left(\bGamma_{c,c'}\bPsi\bz_c^{[k-1]}-\bGamma_{c',c}\bPsi\bz_{c'}^{[k-1]}\right). \label{eq.update3_nu}
\end{align}
\end{subequations}
In \eqref{eq.update3_mu} and \eqref{eq.update3_nu}, both $\bmu_{c,c'}^{[k]}$ and $\bnu_{c',c}^{[k]}$ are updated according to the difference between the reconstructed anchor node observations at caches $c$ and $c'$. Hence, by choosing the same initial values for $\bmu_{c,c'}^{[k]}$ and $\bnu_{c',c}^{[k]}$, i.e., by setting  
$\bmu_{c,c'}^{[0]} = \bnu_{c',c}^{[0]}$, 
we have $\bmu_{c,c'}^{[k]} = \bnu_{c',c}^{[k]}$ for all $k>0$. Consequently, the set of recursive equations in \eqref{eq.update} can be further simplified as
\begin{subequations} \label{eq.update2}
\begin{align}
    \blambda_c^{[k]} &= \blambda_c^{[k-1]} + \rho\left(\bPhi_c\bPsi\bz_c^{[k-1]}-\by_c\right), \label{eq.update2_lambda}\\
    \bmu_{c,c'}^{[k]} &= \bmu_{c,c'}^{[k-1]} + \frac{\rho}{2}\left(\bGamma_{c,c'}\bPsi\bz_c^{[k-1]}\!-\!\bGamma_{c',c}\bPsi\bz_{c'}^{[k-1]}\right),\,\forall c'\!\in\!\calD_c, \label{eq.update2_mu}\\
    \bxi_c^{[k]} &= \bxi_c^{[k-1]} + \rho\left(\bz_c^{[k-1]}-\tbz_c^{[k-1]}\right), \label{eq.update2_xi}\\
    \bz_c^{[k]} &= \frac{1}{\rho}\bigg[ (\bPhi_c\bPsi)^T\bPhi_c\bPsi + 2\sum_{c'\in\calD_c}(\bGamma_{c,c'}\bPsi)^T\bGamma_{c,c'}\bPsi+ {\bf I} \bigg]^{-1}\notag\\
    & \hspace{0.6cm}\cdot\bigg[\left(\bPhi_c\bPsi\right)^T\left(\rho\by_c-\blambda_c^{[k]}\right) + \rho\tbz_c^{[k-1]} - \bxi_c^{[k]}\notag\\    &\hspace{0.6cm}+\rho\sum_{c'\in\calD_c}\left(\bGamma_{c,c'}\bPsi\right)^T\left(\bGamma_{c,c'}\bPsi\bz_c^{[k-1]}\!+\!\bGamma_{c',c}\bPsi\bz_{c'}^{[k-1]}\right)\notag\\    &\hspace{0.6cm}-2\sum_{c'\in\calD_c}\left(\bGamma_{c,c'}\bPsi\right)^T\bmu_{c,c'}^{[k]} \bigg], \label{eq.update2_z}\\
    \tilde z_{c,i}^{[k]} &= {\cal S}_{1/\rho}\bigg(z_{c,i}^{[k]}+\frac{\xi_{c,i}^{[k]}}{\rho}\bigg), ~i = 1, \cdots, NW. \label{eq.update2_tz}
\end{align}
\end{subequations}

Notice that the above update equations are implemented in parallel at the local caches with only the exchange of the recovered anchor node observations with its neighbors in each iteration. The update procedure at cache $c$ is summarized in Algorithm \ref{alg:admm}. Specifically, in
iteration $k$, cache $c\in\calC$ first updates the multipliers $\blambda_c^{[k]}$, $\{\bmu_{c,c'}^{[k]}\}_{c'\in\calD_c}$, $\bxi_c^{[k]}$ by using the measurement $\by_c$, the estimated sparse signal $\bz_c^{[k-1]}$ , the reconstructed anchor node observations of cache $c$ and its neighboring caches in iteration $k-1$, and the auxiliary variable $\tbz_c^{[k-1]}$. Afterwards, the current $\bz_c^{[k]}$ is updated according to \eqref{eq.update2_z} by using the previous auxiliary variable $\tbz_c^{[k-1]}$, the reconstructed anchor node observations of caches $c$ and its neighboring caches in iteration $k-1$, and the updated multipliers $\blambda_c^{[k]}$, $\bmu_{c,c'}^{[k]}$ and $\bxi_c^{[k]}$. 
Then, the auxiliary variable $\tbz_c^{[k]}$ is updated by using $\bz_c^{[k]}$ and $\bxi_c^{[k]}$ as in \eqref{eq.update2_tz}. Finally, cache $c$ updates its anchor node observation $\bGamma_{c,c'}\bPsi\bz_c^{[k]}$ and sends it to the neighboring caches. 
The above procedures are repeated until the norm of the primal and dual residuals \cite{boyd2011distributed}, defined as
\begin{align}
    \big\|{\bf r}^{[k]}\big\|_2^2 =&\sum_{c=1}^{C}\bigg[\big\|\bPhi_c\bPsi\bz_c^{[k]}-\by_c\big\|_2^2 + \big\|\bz_c^{[k]}-\tbz_c^{[k]}\big\|_2^2\notag\\
    &+ \frac{1}{2}\sum_{c'\in\calD_c}\big\|\bGamma_{c,c'}\bPsi\bz_c^{[k]}-\bGamma_{c',c}\bPsi\bz_{c'}^{[k]}\big\|_2^2\bigg]
\end{align}
\begin{align}
\big\|{\bf s}^{[k]}\big\|_2^2 =& \rho^2\sum_{c=1}^{C}\Big\|\tbz_c^{[k-1]}-\tbz_c^{[k]} + \sum_{c'\in\calD_c}(\bGamma_{c,c'}\bPsi)^T(\bGamma_{c,c'}\bPsi)\notag\\
&\cdot(\bz_c^{[k-1]}-\bz_c^{[k]}+\bz_{c'}^{[k-1]}-\bz_{c'}^{[k]})\Big\|_2^2,
\end{align}
are sufficiently small or the maximum number of iterations is reached. The convergence conditions are chosen as
\begin{align}
    \big\|{\bf r}^{[k]}\big\|_2^2\leq\epsilon_{\rm pri}~~\text{and}~~\big\|{\bf s}^{[k]}\big\|_2^2\leq \epsilon_{\rm dual}, \label{eq.stopping_criterion}
\end{align}
where $\epsilon_{\rm pri}>0$ and $\epsilon_{\rm dual}>0$ are the convergence thresholds for the primal and dual feasibility conditions \cite{boyd2011distributed, Xu_2017}. 
Following \cite{he2000alternating} and \cite{boyd2011distributed}, we also adjust the penalty parameter $\rho$ in each iteration to further improve the convergence rate.
In particular, the penalty parameter scheme is adjusted as
\begin{align} \label{eq.penalty_parameter}
    \rho^{[k]} =
\begin{cases}
\tau\rho^{[k-1]}& \text{if } \big\|{\bf r}^{[k-1]}\big\|_2 > \eta\big\|{\bf s}^{[k-1]}\big\|_2,\\
\tau^{-1}\rho^{[k-1]}& \text{if } \big\|{\bf s}^{[k-1]}\big\|_2 > \eta\big\|{\bf r}^{[k-1]}\big\|_2,\\
\rho^{[k-1]} & \text{otherwise},
\end{cases}
\end{align}
where $\tau>1$ and $\eta>1$ are parameters that scale the rate of increase or decrease of $\rho$.
Increasing $\rho^{[k]}$ tends to have larger dual residual ${\bf s}^{[k]}$ but smaller primal residual ${\bf r}^{[k]}$. Conversely, decreasing $\rho^{[k]}$ tends to produce smaller dual residual ${\bf s}^{[k]}$ but larger primal residual ${\bf r}^{[k]}$. 


\begin{algorithm}[t]
  \KwIn{$\by_c\in\mathbb{R}^{M_cW}$.}
  \KwOut{Optimal ${\bz}_c^*, \forall c\in\calC$.}
  \Variable{$\bz_c^{[k]}\in\mathbb{R}^{NW}, \tbz_c^{[k]}\in\mathbb{R}^{NW}, \blambda_c^{[k]}\in\mathbb{R}^{M_cW}, \bmu_{c,c'}^{[k]}\in\mathbb{R}^{|\calQ_{c,c'}|W}, \bxi_c^{[k]}\in\mathbb{R}^{NW}, \forall c'\in\calD_c$ for each cache $c\in\calC$.}
  \Initialization{Set $k = 0, \rho^{[0]} = 10$. All variables are set to ${\bf 0}$.}
  \Repeat{the stopping criterion \eqref{eq.stopping_criterion} is satisfied}{
    $k \leftarrow k + 1$\;
    \For{$\forall c\in\calC$ (in parallel)}{
        Update $\blambda_c^{[k]}$ via \eqref{eq.update2_lambda}\;
        Update $\bmu_{c,c'}^{[k]}$ via \eqref{eq.update2_mu}\;
        Update $\bxi_c^{[k]}$ via \eqref{eq.update2_xi}\;
        Update $\bz_c^{[k]}$ via \eqref{eq.update2_z}\;
        Update $\tbz_c^{[k]}$ via \eqref{eq.update2_tz}\;
        Send $\bGamma_{c,c'}\bPsi\bz_c^{[k]}$ to and receive $\bGamma_{c',c}\bPsi\bz_{c'}^{[k]}$ from all caches $c'\in\calD_c$\;
    }
    Update $\rho^{[k+1]}$ via \eqref{eq.penalty_parameter}\;
    }
  \caption{Collaborative Sparse Recovery by Anchor Alignment (CoSR-AA)}
  \label{alg:admm}
\end{algorithm}

{\bf Complexity Analysis.} Here, we summarize the computational complexity of each iteration by measuring the number of multiplications. 
For \eqref{eq.update2_lambda} and \eqref{eq.update2_mu}, since $\bPhi_c$ and $\bGamma_{c,c'}$ represent the node selection in practice and not part of the actual computation, 
we only consider the complexity of computing 
$\bPsi\bz_c^{[k-1]}$, 
which is $\calO(N^2W^2)$ in both equations. Moreover,  \eqref{eq.update2_xi} and \eqref{eq.update2_tz} involve the multiplication of a scalar $\rho$ on each entry of the vector. Therefore, the complexity is only $\calO(NW)$ for both. For \eqref{eq.update2_z}, the complexity mainly depends on the matrix inversion whose complexity is $\calO(N^3W^3)$. Since the variables are updated for at most $K$ iterations and the computations are executed in parallel by the $C$ caches, the overall complexity of the CoSR-AA algorithm is $\calO(KN^3W^3)$.


\section{Deep Unfolding of the CoSR-AA Algorithm} \label{sec.deep_unfolding}

While the ADMM-based COSR-AA algorithm proposed in the previous section enables collaborative sparse recovery of sensor data at different caches, the number of iterations required to arrive at a consensus solution may be large. To reduce the computation time, in this section we propose a distributed recovery algorithm based on the deep unfolding \cite{yang_sun_li_xu_2017} of the proposed CoSR-AA algorithm, named Deep CoSR-AA. The key idea is to treat the operations of each ADMM iteration as a separate layer in a deep neural network and utilize the end-to-end training of their parameters to enable reliable recovery at the output, even when using a small number of layers (or iterations).

Specifically, the proposed deep learning model takes the form of a graph neural network where, in iteration $k$, the local network at cache $c$ takes as input the previous local multiplier variables $\blambda_c^{[k-1]}$, $\{\bmu_{c,c'}^{[k-1]}\}_{c'\in\calD_c}$, $\bxi_c^{[k-1]}$, the reconstructed sparse vectors $\bz_c^{[k-1]}$ and $\tbz_c^{[k-1]}$, the messages $\{\bu_{c',c}^{[k-1]}\}_{c'\in\calD_c}$ received from its neighboring caches, and the messages $\{\bu_{c,c'}^{[k-1]}\}_{c'\in\calD_c}$ computed locally in the previous iteration, and produces as output the updated variables $\blambda_c^{[k]}$, $\{\bmu_{c,c'}^{[k]}\}_{c'\in\calD_c}$, $\bxi_c^{[k]}$, $\bz_c^{[k]}$, $\tbz_c^{[k]}$ and the updated messages $\{\bu_{c,c'}^{[k]}\}_{c'\in\calD_c}$. Notice that $\{\bu_{c,c'}^{[k]}\}_{c'\in\calD_c}$ are messages that are sent to neighboring caches to facilitate collaboration and ensure consistency among multiple caches. In the original CoSR-AA algorithm, the message sent from cache $c'$ to cache $c$ takes on the form $\bu_{c',c}^{[k-1]}=\bGamma_{c',c}\bPsi\bz_{c'}^{[k-1]}$, which corresponds to the reconstructed observations of the anchor nodes chosen by caches $c'$ and $c$. The exchange of the reconstructed anchor node observations $\bu_{c',c}^{[k-1]}$, instead of the full observation vector $\bx_{c'}^{[k-1]}=\bPsi\bz_{c'}^{[k-1]}$, reduces the messaging overhead between caches $c'$ and $c$ since the dimension of $\bu_{c',c}^{[k-1]}$ is much smaller than $\bx_{c'}^{[k-1]}$. By adopting a deep learning framework, we can generalize the dimension reduction by utilizing an encoder network to better extract the necessary features that can be exchanged with other caches.
The proposed graph neural network model consists of a multiplier update subnetwork, an aggregation subnetwork, a reconstruction subnetwork, and a message embedding layer at each node. Details of the neural network architecture are shown in Fig. \ref{fig:nn_structure_stage_k} and described below.

\begin{figure}[t]
    \centering
    \includegraphics[width=1.0\linewidth]{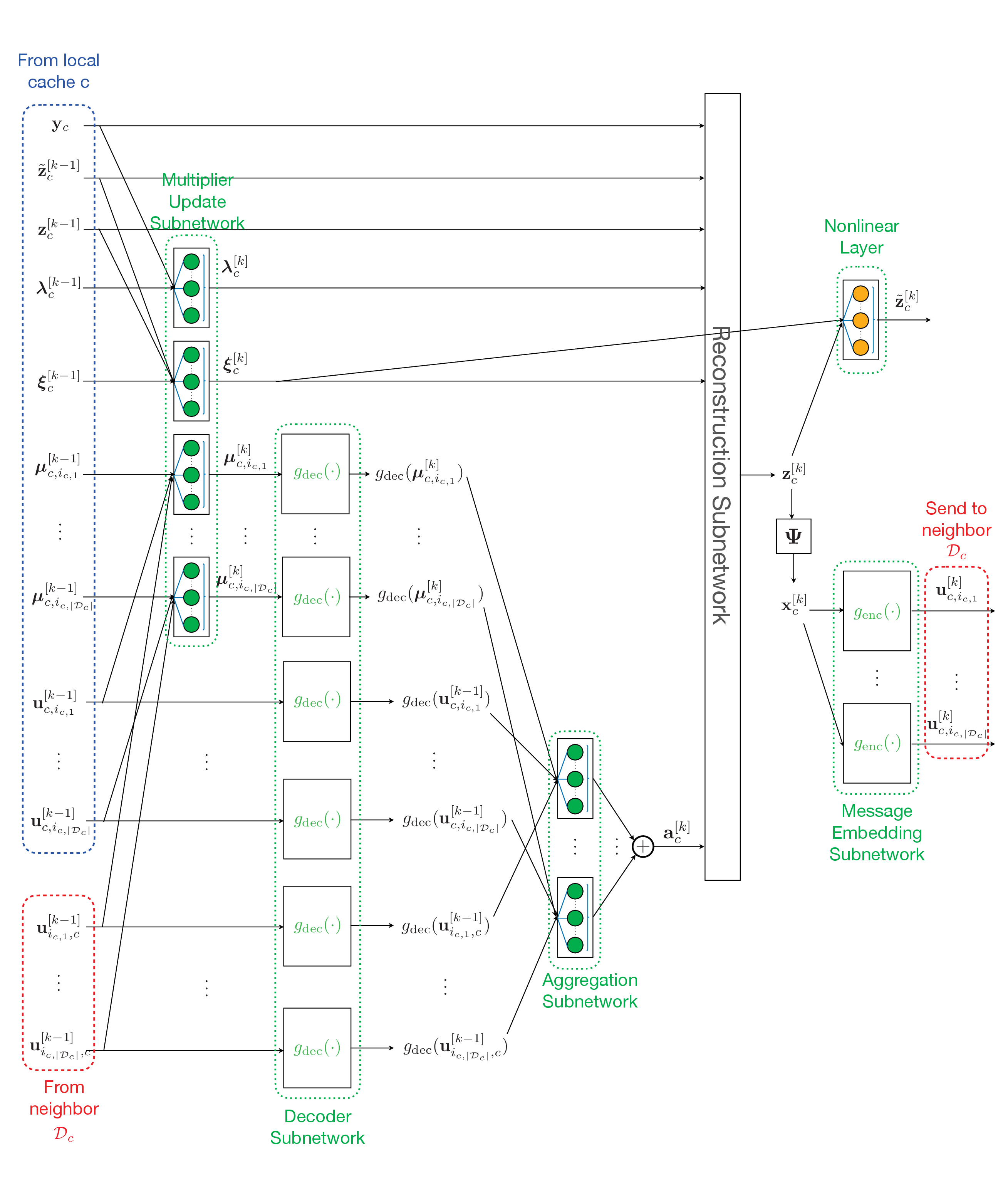}
    \caption{The neural network architecture at cache $c$ in stage $k$ of the Deep CoSR-AA algorithm. The green circle and the orange circle denote the neuron without and with the activation function of the neural network, respectively.}
    \label{fig:nn_structure_stage_k}
\end{figure}

\textbf{Multiplier Update Subnetwork:} Following \eqref{eq.update2_lambda}, \eqref{eq.update2_mu}, 
and \eqref{eq.update2_xi}, the multiplier update subnetwork takes the previous multiplier variables $\blambda_c^{[k-1]}$, $\{\bmu_{c,c'}^{[k-1]}\}_{c'\in\calD_c}$, $\bxi_c^{[k-1]}$, reconstructed vectors $\bz_c^{[k-1]}$, $\tbz_c^{[k-1]}$, and messages $\{\bu_{c,c'}^{[k-1]}\}_{c'\in\calD_c}$ and $\{\bu_{c',c}^{[k-1]}\}_{c'\in\calD_c}$ as input and updates the multiplier variables by passing them through a linear layer with operations
given by 
\begin{subequations} \label{eq.nn_structure_dual_variable}
\begin{align}
    \blambda_c^{[k]} =& \bW_{\blambda_c}^{[k]}\Big[    \blambda_c^{{[k-1]}^{T}}, \big(\bPhi_c\bPsi\bz_c^{[k-1]}\big)^{\!T}, \by_c^T \Big]^T \!\!+ \bb_{\blambda_c}^{[k]},\label{eq.nn_lambda}\\
    \bmu_{c,c'}^{[k]} =& \bW_{\bmu_{c,c'}}^{[k]}\Big[
        \bmu_{c,c'}^{{[k-1]}^{T}}, \bu_{c,c'}^{{[k-1]}^{T}}, \bu_{c',c}^{{[k-1]}^{T}}\Big]^T\!\!+ \bb_{\bmu_{c,c'}}^{[k]}, \notag\\
    &\hspace{5cm}~\forall c'\in\calD_c,\label{eq.nn_mu}\\
    \bxi_c^{[k]} =& \bW_{\bxi_c}^{[k]}\Big[
        \bxi_c^{{[k-1]}^T}, \bz_c^{{[k-1]}^T}, \tbz_c^{{[k-1]}^{\substack{T\\[1.25ex]}}}
    \Big]^T \!\!+ \bb_{\bxi_c}^{[k]},\label{eq.nn_xi}
\end{align}
\end{subequations}
where weight matrices $\bW_{\blambda_c}^{[k]}$, $\bW_{\bmu_{c,c'}}^{[k]}$, $\bW_{\bxi_c}^{[k]}$ and biases $\bb_{\blambda_c}^{[k]}$, $\bb_{\bmu_{c,c'}}^{[k]}$, $\bb_{\bxi_c}^{[k]}$ are learnable variables in neural subnetwork.
Notice that \eqref{eq.update2_lambda}, \eqref{eq.update2_mu}, 
and \eqref{eq.update2_xi} are special cases of \eqref{eq.nn_lambda}, \eqref{eq.nn_mu}, and \eqref{eq.nn_xi} with the weight matrices chosen as $\bW_{\blambda_c}^{[k]}=\begin{bmatrix}
    {\bf I}_{M_cW}, \rho{\bf I}_{ M_cW}, -\rho{\bf I}_{ M_cW}
\end{bmatrix}$, $\bW_{\bmu_{c,c'}}^{[k]}=\begin{bmatrix}
    {\bf I}_{|\calQ_{c,c'}|W},(\rho/2){\bf I}_{|\calQ_{c,c'}|W},-(\rho/2){\bf I}_{|\calQ_{c,c'}|W}
\end{bmatrix}$ and $\bW_{\bxi_c}^{[k]}=\begin{bmatrix}
    {\bf I}_{NW},\rho{\bf I}_{NW},-\rho{\bf I}_{NW}
\end{bmatrix}$ and the biases being zero.

{\bf Aggregation Subnetwork:} In the aggregation subnetwork, the low-dimensional messages sent to and received from all caches in $\calD_c$, i.e.,  $\{\bu_{c,c'}^{[k-1]}\}_{c'\in\calD_c}$ and $\{\bu_{c',c}^{[k-1]}\}_{c'\in\calD_c}$, and the updated multiplier variables $\{\bmu_{c,c'}^{[k]}\}_{c'\in\calD_c}$, each with dimension $|\calQ_{c,c'}|W$, are first mapped back to the ambient dimension $NW$. Inside \eqref{eq.update2_z}, this is done by pre-multiplying the above variables with the matrix $\bGamma_{c,c'}^T$. However, by adopting a deep learning framework, this can be done by a general decoder network $g_{\rm dec}(\cdot)$ to produce the $NW$-dimensional vectors $\{g_{\rm dec}(\bu_{c,c'}^{[k-1]})\}_{c'\in\calD_c}$, $\{g_{\rm dec}(\bu_{c',c}^{[k-1]})\}_{c'\in\calD_c}$, and $\{g_{\rm dec}(\bmu_{c,c'}^{[k]})\}_{c'\in\calD_c}$. The decoded information associated with each $c'$ is aggregated by the function $g_{\rm agg}(\cdot)$ and the aggregated outputs of all neighboring caches are pooled together to yield the subnetwork output 
\begin{equation}
\ba_c^{[k]} = \sum_{c'\in\calD_c} g_{\rm agg}\left(g_{\rm dec}\big(\bu_{c,c'}^{[k-1]}\big), g_{\rm dec}\big(\bu_{c',c}^{[k-1]}\big),g_{\rm dec}\big(\bmu_{c,c'}^{[k]}\big)\right) .\label{eq.nn_agg}
\end{equation}
We shall evaluate the effectiveness of both decoders in the simulations section.

\begin{figure*}[t]
    \centering
    \includegraphics[width=\textwidth,height=\textheight,keepaspectratio]{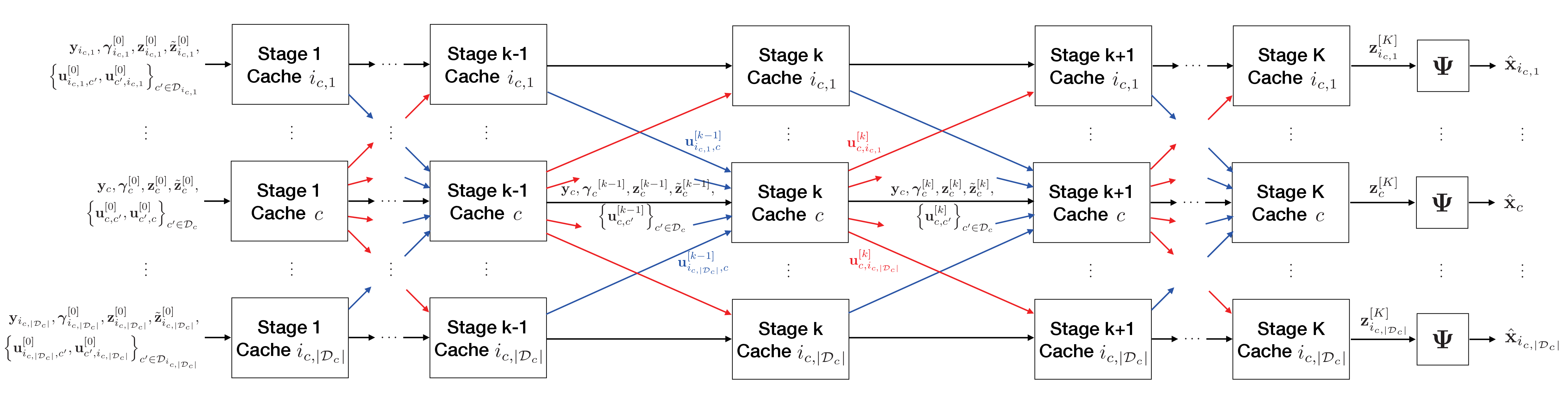}
    \caption{Illustration of the messaging between caches over different stages of the Deep CoSR-AA algorithm. The red arrows indicate the flow of variables sent outward from cache c to its neighboring caches in $\calD_c$. The blue arrows indicate the flow of variables coming inward to cache $c$ from all caches $\calD_c$.}
    \label{fig:nn_sturcture}
\end{figure*}

\textbf{Reconstruction Subnetwork:} The reconstruction subnetwork aims to reconstruct the local sparse vectors $\bz_c^{[k]}$ and $\tbz_c^{[k]}$. Following \eqref{eq.update2_z}, the reconstruction of $\bz_c^{[k]}$ in the subnetwork is implemented by a single linear layer which takes as input the aggregation subnetwork output $\ba_c^{[k]}$, the local observation $\by_c$, the updated multipliers $\blambda_c^{[k]}$ and $\bxi_c^{[k]}$, and the sparse vector estimate $\tilde \bz_c^{[k-1]}$ from the previous iteration to yield the updated reconstruction of the sparse vector
\begin{align}
    \bz_c^{[k]} =& \bW_{\bz_c}^{[k]}\begin{bmatrix}
        \by_c^T, \blambda_c^{{[k]}^{\substack{T\\[0.5ex]}}}, \tbz_c^{{[k-1]}^{\substack{T\\[1.5ex]}}}, \bxi_c^{{[k]}^{\substack{T\\[0.5ex]}}},{\bf a}_{c}^{{[k]}^{\substack{T\\[0.25ex]}}}
\end{bmatrix}^T\!+\bb_{\bz_c}^{[k]},\label{eq.nn_z}
\end{align}
where weight matrices $\bW_{\bz_c}^{[k]}$ and bias $\bb_{\bz_c}^{[k]}$ are learnable variables in neural network. Then, following \eqref{eq.update2_tz}, the updated reconstruction $\bz_c^{[k]}$ is then passed through a nonlinear layer to produce the auxiliary reconstruction vector
\begin{align}
    \tbz_c^{[k]} = \sigma\left(\bW_{\tbz_c}^{[k]}\begin{bmatrix}
       {\bxi_c^{[k]}}^{\substack{T\\[0.5ex]}}, {\bz_c^{[k]}}^T
    \end{bmatrix}^T + \bb_{\bxi_c}^{[k]}\right),\label{eq.nn_tz}
\end{align}
where the weight matrix $\bW_{\tbz_c}^{[k]}$ and bias $\bb_{\tbz_c}^{[k]}$ are learnable variables and $\sigma(\cdot)$ represents the nonlinear activation function. Recall that $\tbz_c^{[k]}$ is an auxiliary variable defined to account for the $\ell_{1}$-norm minimization in \eqref{eq.admm}. 


\textbf{Message Embedding Subnetwork:} The message embedding subnetwork produces the message $\bu_{c,c'}^{[k]}$ that cache $c$ exchanges with cache $c'$ in each stage (say, stage $k$). The message $\bu_{c,c'}^{[k]}$ corresponds to estimates of the observations at the anchor nodes, i.e., $\bGamma_{c,c'}\bPsi\bz_c^{[k]}$, which reduces the dimension from $NW$ to $|\calQ_{c,c'}|W$. However, by adopting a deep learning framework, it is possible to consider a more general encoder network $g_{\rm enc}(\cdot)$ to extract the necessary features for exchange with neighboring caches. 
In particular, the message embedding subnetwork can be expressed as
\begin{align}
   \bu_{c,c'}^{[k]} = g_{\rm enc}\left(\bPsi\bz_c^{[k]}\right), \forall c'\in\calD_c,\label{eq.nn_u}
\end{align}
where the low-dimensional output $\bu_{c,c'}^{[k]}$ is sent from cache $c$ to cache $c'\in\calD_c$. Notice that a common encoder $g_{\rm enc}$ is used to encode messages destined for different neighbors. In general, different encoders can be used, but this
may increase the model complexity and limit the scalability of the proposed method.
In the simulations, we shall consider two choices of $g_{\rm enc}$: 
(i) CoSR-AA-Linear (CoSR-AA-LI): the case where $g_{\rm enc}$ is a simple linear layer, i.e., $g_{\rm enc}(\bx)=\bW_{\rm enc}\bx+{\bf b}_{\rm enc}$, and (ii) CoSR-AA-Autoencoder (CoSR-AA-AE): the case where $g_{\rm enc}$ is a general neural network with nonlinear activation. In the latter case, $g_{\rm enc}$ can be trained together with $g_{\rm dec}$ in the aggregation subnetwork to form an autoencoder-like network.

The subnetworks described above form the computations at cach $c$ in stage $k$. The messaging between caches from one stage to the next is illustrated in Fig. \ref{fig:nn_sturcture}.
Here, $\by_c, \bgamma_c^{[0]}, \tbz_c^{[0]}$ and $\bz_c^{[0]}$ are the initial input values. 
In stage $k$, cache $c$ takes as input both the local output at stage $k-1$ and the messages $\bu_{c',c}^{[k-1]}$ received from its neighbor $c'\in\calD_c$. 
It then produces at the output $\bgamma_c^{[k]}$, $\tbz_c^{[k]}$, $\bz_c^{[k]}$ and $\{\bu_{c,c'}^{[k]}\}_{c'\in\calD_c}$, which are passed to the next stage.
The messages $\{\bu_{c,c'}^{[k]}\}_{c'\in\calD_c}$ are also sent to its neighboring caches in $\calD_c$. 
After $K'$ stages, the output $\bz_c^{[K']}$ is multiplied by $\bPsi$ to obtain the reconstructed sensor data $\hat{\bx}_c$.

The deep learning model described above is trained end-to-end using a multi-stage loss function that considers the sum of the mean square error (MSE) loss of all stages. This is often adopted in the literature, e.g., \cite{ma_jiang_liu_ma_2022}, to avoid the gradient vanishing problem. 
The loss function is defined as 
\begin{align} \label{eq.loss_function}
    {\rm Loss} =& \sum_{c=1}^{C}\frac{1}{|{\cal E}|}\sum_{i\in{\cal E}}\sum_{k=1}^{K}\Big\|\hat{\bx}_c^{[k]}(i)-\bx(i)\Big\|_2^2+ \epsilon \cdot {\rm reg}({\cal G}),
\end{align}
where ${\cal E}$ is the training data set and $\cal G$ is the set of all model parameters. Here, ${\rm reg}({\cal G})$ represents the $\ell_2$-regularization of the model parameters in $\cal G$.


{\bf Complexity Analysis.} Here, we discuss the computational complexity required by the Deep CoSR-AA algorithm in terms of the number of multiplications. 
Notice that, in the proposed neural network model, the number of multiplications scales with the input dimension which is $\calO(N^2W^2)$ (see, e.g., 
\eqref{eq.nn_z}) and avoids
the need for the matrix inversion in \eqref{eq.update2_z}, which originally requires $\calO(N^3W^3)$ .
Suppose that $K'$ stages are adopted in the Deep CoSR-AA method. In this case, the overall complexity scales only as $\calO(K'N^2W^2)$, which is significantly smaller than that of the CoSR-AA algorithm. More importantly, the required number of stages $K'$ in the Deep CoSR-AA method is much smaller than the number of iterations $K$ required in the original CoSR-AA algorithm (by at least an order of magnitude, as demonstrated in our experiments).
More discussions are given in Section \ref{sec.simulations}. 

\section{Simulation Results} \label{sec.simulations} 


In this section, numerical simulations are provided to demonstrate the effectiveness of the proposed collaborative sensor caching and sparse data recovery algorithms. 
In our experiments, unless mentioned otherwise, we deploy $N=100$ sensors and $C=4$ caches over a $100\sqrt{N} \times 100\sqrt{N}$ square region. The sensor field is divided into $N$ blocks, each with area $100 \times 100$, and a sensor is deployed randomly according to the uniform distribution in each block. The sensor field is also equally divided into $C=4$ square subregions with a cache deployed at the center of each subregion. Each cache can only access data from sensors located within its subregion, and thus the subregion is referred to as the cache's coverage. 

Following \cite{leinonen2015sequential}, we generate each sensor's observation by taking an aggregate of $S = 10$ underlying Gaussian sources distributed randomly according to the uniform distribution throughout the field. Hence, sensor $n$'s observation at time $t$ can be written as $x_n(t)=\sum_{s=1}^{S} e^{-(d_{n,s}/\eta_1)^{2}}\beta_s(t)$, 
where 
$\eta_1>0$ is the correlation length, $d_{n,s}$ is the distance between sensor $n$ and source $s$, and $\beta_s(t)$ is the value of the $s$-th Gaussian source at time $t$. Here, we set $\eta_1=800$ unless otherwise specified.
Moreover, to incorporate temporal correlation, we choose the source sequence $\boldsymbol{\beta}_s = [\beta_s(1),\ldots, \beta_s(T)]^T$ to be the sum of a smooth component $\boldsymbol{\lambda}_s = [\lambda_s(1),\ldots, \lambda_s(T)]^T$ and a non-smooth component $\boldsymbol{\pi}_s = [\pi_s(1),\ldots, \pi_s(T)]^T$, i.e., $\boldsymbol{\beta}_s = \boldsymbol{\lambda}_s+\boldsymbol{\pi}_s$ \cite{leinonen2015sequential}. The sequence $\blambda_s$ is generated by a Gauss-Markov process where $\lambda_s(t) = \alpha_s(\lambda_s(t-1)-\mu_s) + \sqrt{1-\alpha_s^2}\nu_s(t) + \mu_s$ with $\alpha_s=0.9$ being the correlation coefficient and $\nu_s(t)\sim\calN(0, 1)$ being a Gaussian innovation term that is i.i.d. over time. 
The mean $\mu_s$ is chosen as $\lambda_s(1)$. The sequence $\boldsymbol{\lambda}_s$ is further smoothened by preserving only the low-pass components. The non-smooth component $\boldsymbol{\pi}_s$ is generated by a Markov chain with state space
${\cal V}_s = \{V_{s,1},\ldots, V_{s,|{\cal V}_s|}\}$. Here, we choose $|{\cal V}_s|=10$, for all $s$, and let $V_{s,j}\sim\calN(0, 1)$, for all $j$. The self-transition probability of each state is $0.8$ and the transition probability to any other state is $0.2/|{\cal V}_s|$.
Further details on the sensor data generation can be found in \cite{leinonen2015sequential}.
Moreover, following \cite{duarte2012kronecker} and \cite{leinonen2015sequential}, we choose  $\bPsi_S\in\mathbb{R}^{N\times N}$ in the spatial domain to be a 2D-DCT matrix 
and $\bPsi_T\in\mathbb{R}^{W\times W}$ in the temporal domain to be 
an inverse DCT matrix. The number of measurements is fixed over time and the same for all caches, i.e., $M_c = M$ for all $c$. The number of anchor nodes chosen by different pairs of caches is also assumed to be the same, i.e., $|\calQ_{c,c'}| = Q$, $\forall c'\in\calD_c, c\in\calC$, and the window size is $W=4$. The quality of sensor data recovery is evaluated using the normalized mean square error (NMSE), defined as
\begin{align} \label{eq.WinDataNMSE}
    {\rm NMSE} = \frac{1}{C(T-W+1)} \sum_{c=1}^{C} \sum_{t=W}^{T} \frac{\left\|\hat{\bX}_c(t) - \bX(t)\right\|_F^2}{\left\|\bX(t)\right\|_F^2},
\end{align}
where $\hat{\bX}_c(t)$ is the estimated recovery of $\bX(t)$ at cache $c$. All experimental results are averaged over $10$ random deployment scenarios. 

\subsection{Experiments for the Proposed CoSR-AA Algorithm} \label{subsec.prop_CoSR-AA}

\begin{figure}[t]
    \centering
    \includegraphics[width=1.0\linewidth]{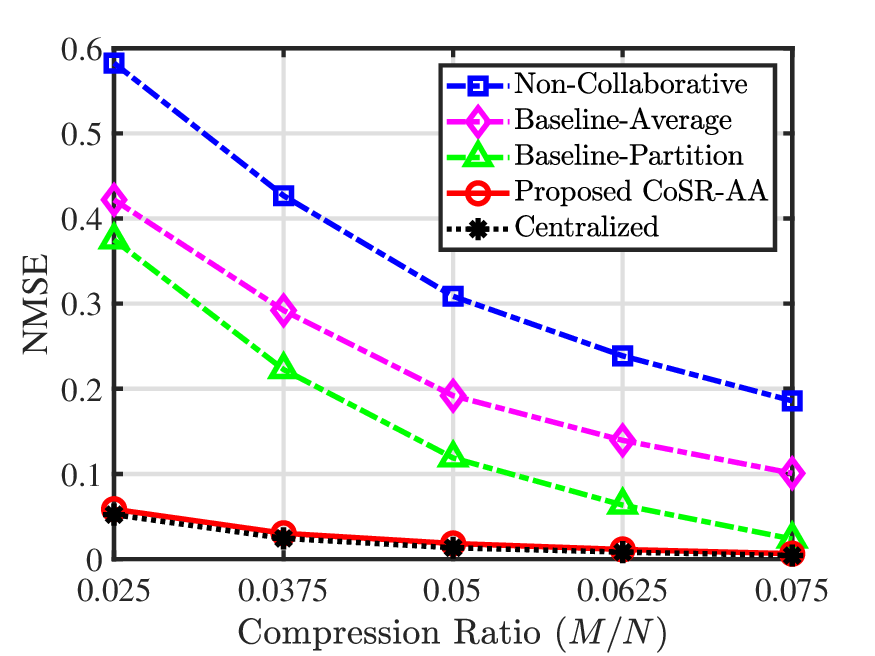}
    \caption{NMSE versus the compression ratio $M/N$ for the case with $C = 4$.}
    \label{fig:Co_vs_NoCo_C=4_W=4}
    \vspace{-.2cm}
\end{figure}

We first examine the effectiveness of the proposed ADMM-based CoSR-AA algorithm. The experiments are conducted in Matlab. 
The maximum number of iterations is $K=3000$. The convergence thresholds in \eqref{eq.stopping_criterion} are chosen as $\epsilon_{\rm pri} = 0.004$ and $\epsilon_{\rm dual} = 1.5$, and the parameters used in \eqref{eq.penalty_parameter} for adaptability are $\tau = 2$ and $\eta = 10$. The initial value of the penalty parameter is given by $\rho^{[0]} = 10$.  In addition to the ideal centralized CS-based recovery method in \eqref{eq.centralized_reconstruct}, we also compare with the following three baseline distributed sparse recovery methods:
\begin{itemize}
    \item Non-Collaborative Recovery: In this method, the data associated with the entire sensor field are recovered separately at each cache by solving 
    the problem in \eqref{eq.distributed_reconstruct_with_AN} without the pairwise consensus constraints in \eqref{eq.anchor_constraint}. Without cooperation among caches, the recovery may suffer from insufficient measurements stored at each cache.

    \item Baseline-Average: This method takes the average of the non-collaborative solutions obtained at all caches following the method mentioned above. 
    That is, we have $\Bar{\bX}(t) = \frac{1}{C}\sum_{c=1}^{C}\hat{\bX}_c(t)$.  
    Note that this method requires the caches to exchange the entire local reconstructed sensor field with each other.  
    
    \item Baseline-Partition: This method is also based on the non-collaborative recovery method but, instead of taking the average of the local solutions, we take from each cache only the recovered data associated with its coverage and patch them together to form an estimate of the entire sensor field.
    That is, we take from cache $c$ the recovered data associated with the sensors in $\calN_c$.

\end{itemize}

In Fig. \ref{fig:Co_vs_NoCo_C=4_W=4}, we show the NMSE of all methods with respect to different values of the compression ratio $M/N$. Here, we set the number of anchors as $Q=25$. The anchor nodes associated with each pair of caches are randomly selected from the union of their coverage regions. 
As expected, the NMSE of all methods decreases as the $M/N$ increases since more observations will be available at the caches for data recovery. We can also see that
all the collaborative reconstruction schemes have lower NMSE as compared to the non-collaborative recovery method. In fact, the proposed CoSR-AA approach outperforms all baseline recovery
methods and can maintain good recovery performance even at a low compression ratio (e.g., $ M/N<0.0375$). This is because our proposed scheme exploits anchor alignment to facilitate a consensus solution among caches within the network, enabling low data reconstruction error at local caches even with only a small set of measurements. 
In fact, the proposed CoSR-AA can achieve an NMSE that is close to that of the centralized solution, even though the consensus constraints are limited to only anchor nodes rather than all sensors in the field.


\begin{figure}[t]
    \centering
    \includegraphics[width=1.0\linewidth]{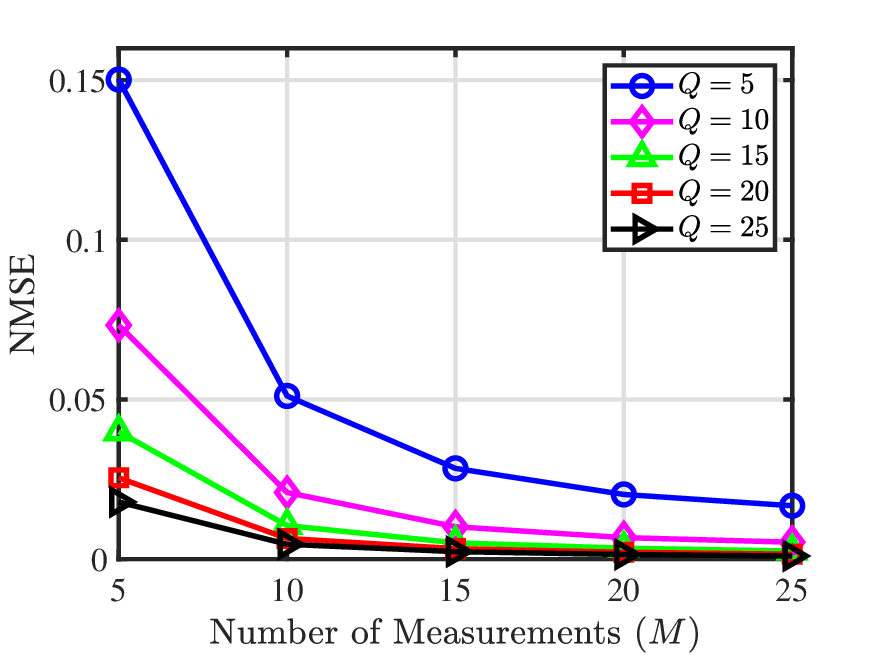}
    \caption{NMSE versus the number of measurements $M$ for different values of $Q$.}
    \label{fig:sample_node_vs_anchor_node}
    \vspace{-.2cm}
\end{figure}

To further investigate the impact of the number of sample nodes $M$ and the number of anchor nodes $Q$ on the data reconstruction performance,
we show in Fig. \ref{fig:sample_node_vs_anchor_node} the NMSE of the proposed CoSR-AA scheme with respect to $M$, for cases with $Q=5$, $10$, $15$, $20$, and $25$.   
We can see that the data reconstruction error decreases with $M$ and $Q$ under all schemes. 
More interestingly, we observe that a similar NMSE can be achieved between cases with $(M,Q) = (5,25)$ and $(M,Q) = (25,5)$ (and also between cases $(M,Q) = (10,10)$ and $(M,Q) = (20,5)$). This implies that, even with only a few local measurements, each cache can still achieve stable data recovery, by aligning their reconstruction at a sufficient number of anchor nodes, enabling each cache to better utilize the measurements available at other caches.
However, if the number of local measurements is already sufficient to achieve good recovery performance, then only a few anchor nodes will be required to ensure consistency among neighboring caches.     

\begin{figure}[t]
    \centering
    \includegraphics[width=1.0\linewidth]{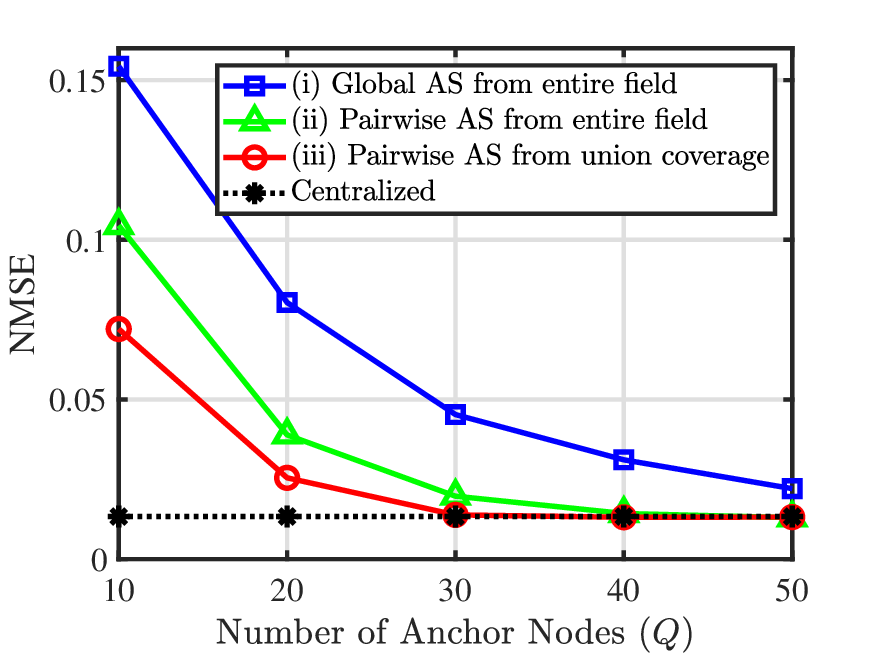}
    \caption{NMSE versus the number of anchor nodes $Q$ for different anchor selection strategies.}
    \label{fig:AN_selection}
    \vspace{-.2cm}
\end{figure}



Next, we examine the performance of the proposed CoSR-AA algorithm with different anchor selection (AS) strategies. In this experiment, $C=4$ caches are deployed in the sensing field with disjoint coverage areas (thus, $|\calN_c|=N/C=25,\forall c\in\calC$). We consider three possible AS strategies: (i) global AS where anchor nodes are selected randomly from the entire sensor field and the same nodes are used by all pairs of caches; 
(ii) pairwise AS with the anchor nodes are again selected from the entire sensor field but are chosen separately for each pair of caches; 
(iii) pairwise AS with anchor nodes chosen only from the union of their coverages. 

In Fig. \ref{fig:AN_selection}, we show the NMSE of the proposed CoSR-AA algorithm with respect to the number of anchor nodes $Q$ for the different AS strategies mentioned above. The centralized solution is also shown for reference. We can see that the NMSE decreases with the number of anchor nodes in all cases. However, the pairwise AS strategies in (ii) and (iii) are able to achieve much lower NMSEs compared to the global AS strategy in (i) since pairwise AS strategies enjoy more diversity and thus avoid errors caused by a few unreliable anchor nodes. We can also see that (iii) outperforms (ii) since the recovery of data associated with sensors outside of their coverages may be unreliable and, thus, may not be suitable 
serving as anchor nodes. 
With the aid of anchor alignment, we can further observe from Fig. \ref{fig:AN_selection} that the proposed anchor selection schemes (ii) and (iii) can achieve the same NMSE as the centralized solution, provided that the number of anchor node is large enough such that $Q\geq 40$ and $Q\geq 30$, respectively. This result confirms that the proposed anchor alignment schemes can indeed yield reliable data recovery while significantly reducing the communication overhead.

\begin{figure}[t]
    \centering
    \includegraphics[width=1.0\linewidth]{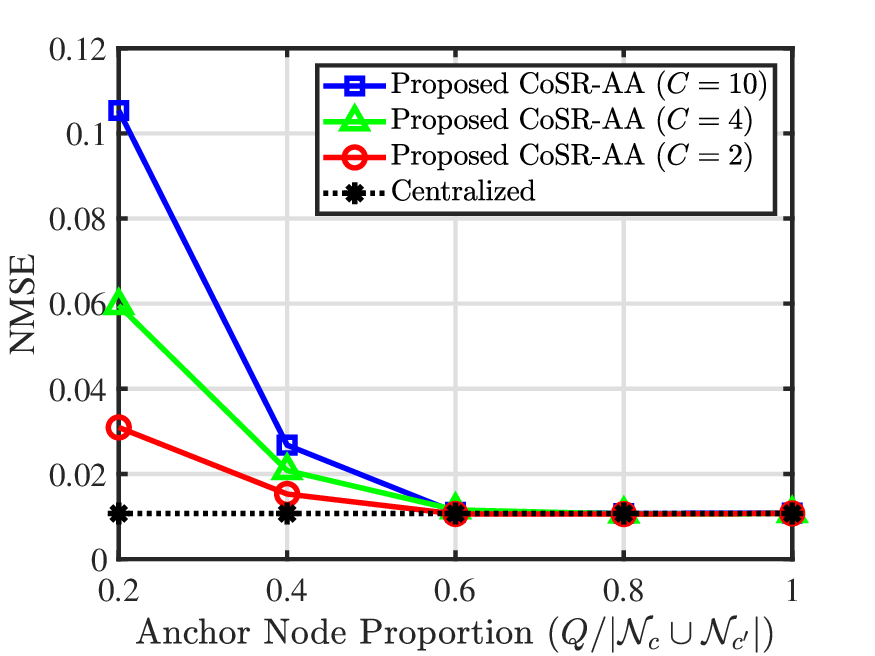}
    \caption{NMSE versus the anchor node proportion for cases with $C=2$, $4$, and $10$.} 
    \label{fig:distributed_vs_centralized}
    \vspace{-.2cm}
\end{figure}

In Fig. \ref{fig:distributed_vs_centralized}, we consider the pairwise AS strategy in (iii) and further examine the effectiveness of the proposed CoSR-AA algorithm for varying proportions of anchor nodes with each pairwise union coverage, i.e., $Q/|\calN_c \cup \calN_{c'}|$. 
In this experiment, the number of caches is set as $C=2$, $4$, and $10$, and the total number of measurements is $M_{\rm tot} = 20$. Each cache collaborates with every other cache. We again choose the centralized solution as our baseline scheme. We can see that the NMSE decreases as the proportion of anchor nodes increases in all cases, even performing as well as the centralized case when $Q/|\calN_c \cup \calN_{c'}| \geq 0.6$. 
This is expected since the measurement size ($M=2$) in the case $C=10$ is less than that in other cases (e.g., $M=10$ for $C=2$ and $M=5$ for $C=4$), implying that more anchor nodes should be employed and more messages should be exchanged to better exploit collaboration among caches. 
For a fixed proportion of anchor nodes, namely, $Q/|\calN_c\cup\calN_{c'}| = 0.2$, the case with $C=2$ 
yields the best performance since the number of measurements (i.e., $M_{\rm tot}/C=10$) available at each local cache is already sufficient to ensure reliable recovery.

\subsection{Experiments for the Proposed Deep CoSR-AA Algorithm} \label{subsec.prop_deep_CoSR-AA}

Next, we examine the effectiveness of the proposed Deep CoSR-AA algorithm using PyTorch. 
The data set is produced by considering $400$ sensor deployment scenarios with each producing data over $125$ time instants. For each deployment, $80$ data points are used for training, $20$ for validation, and $25$ for testing within each deployment. The total number of training data, validation data, and testing data is $32000$, $8000$, and $10000$, respectively. We shuffle the training data to prevent the occurrence of any bias during the training process. We choose Adam for the optimizer, Leaky ReLU for the activation function, and set batch size $|{\cal E}| = 500$. The regularization parameter is set as $\epsilon = 0.001$. For each case, the model is trained over $500$ epochs. The learning rate for Deep CoSR-AA-LI is set as $10^{-4}$ for the first $300$ epochs and subsequently adjusted to $10^{-5}$ for the remaining $200$ epochs. The learning rate for Deep CoSR-AA-AE is instead fixed as $5\times 10^{-4}$. We compare our approach with 
ADMM-CSNet \cite{yang_sun_li_xu_2020},
where only the $\rho$'s in \eqref{eq.update2_lambda}, \eqref{eq.update2_mu}, \eqref{eq.update2_xi} and \eqref{eq.update2_z} were treated as training parameters and a piecewise linear function was used to approximate the nonlinear equation in \eqref{eq.update2_tz}.

\begin{figure}[t]
    \centering
    \includegraphics[width=1.0\linewidth]{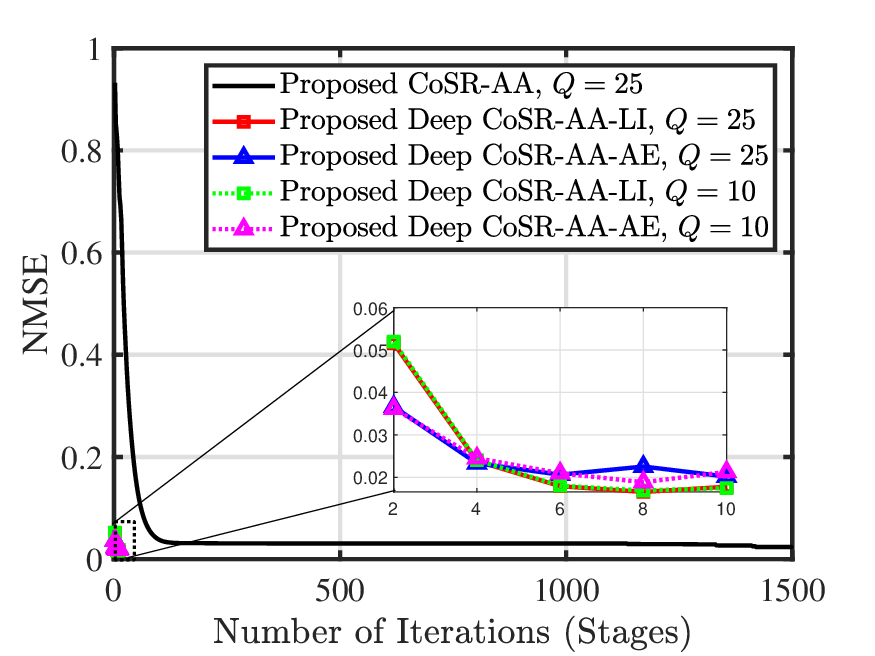}
    \caption{NMSE versus the number of iterations (or stages) for the CoSR-AA and the Deep CoSR-AA methods.}
    \label{fig:admm_vs_unfolding}
    \vspace{-.2cm}
\end{figure}

In this experiment, we set $C=2$ and $M=10$. The same multiplier update subnetwork, aggregation subnetwork, reconstruction subnetwork and message embedding subnetwork are utilized for all caches.
As mentioned in Section \ref{sec.deep_unfolding}, we consider two choices of $g_{\rm enc}$: (i) CoSR-AA-Linear (CoSR-AA-LI), where $g_{\rm enc}$ is a simple linear layer and (ii) CoSR-AA-Autoencoder (CoSR-AA-AE), where $g_{\rm enc}$ is a general neural network with nonlinear activation. In CoSR-AA-AE, the encoder subnetwork $g_{\rm enc}$ contains two nonlinear layers with leaky ReLU as the activation functions plus a linear layer to expand the output into the desired dimension $QW$. The output dimensions of the nonlinear layers are $250$ and $150$, respectively. The decoder subnetwork $g_{\rm dec}$ contains two nonlinear layers plus a linear layer whose output dimensions are $250$, $350$ and $NW$, respectively.

In Fig. \ref{fig:admm_vs_unfolding}, we compare the NMSE of the proposed CoSR-AA and Deep CoSR-AA algorithms with respect to the number of iterations (or stages), for cases with $Q=25$ and $10$. We can see the NMSE decreases as the number of iterations (or stages) increases for all schemes. 
The proposed Deep CoSR-AA schemes are able to achieve low NMSE even with only a few stages ($<10$) whereas the original CoSR-AA algorithm may require a large number of iterations ($\approx 1500$) to converge.

\begin{figure}[t]
    \centering
    \includegraphics[width=1.0\linewidth]{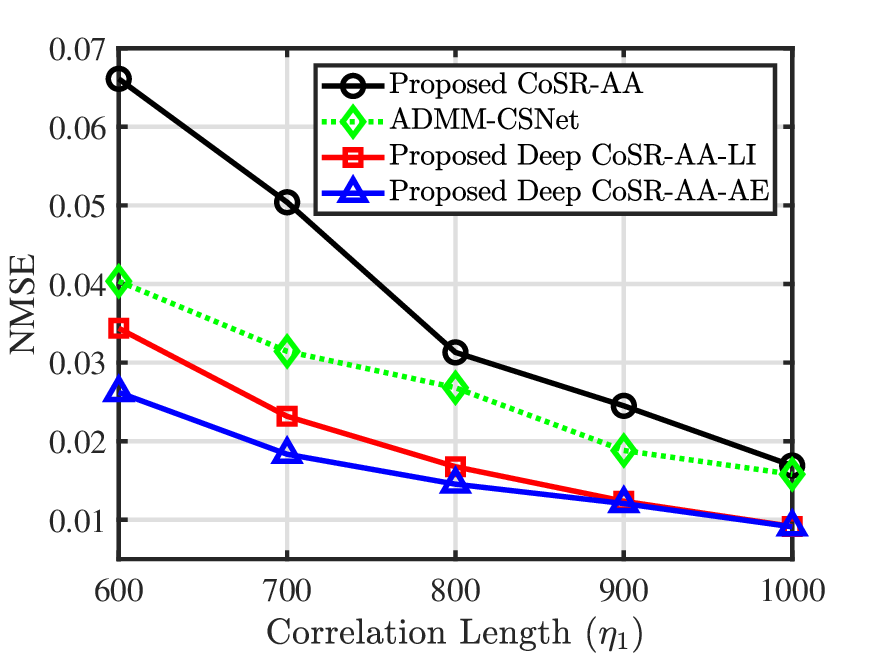}
    \caption{NMSE versus the correlation length for the case where $C=2$ and $M=10$.}
    \label{fig:different_unfolding}
    \vspace{-.2cm}
\end{figure}

In Fig. \ref{fig:different_unfolding}, we show the NMSE versus the correlation length $\eta_1$. A larger $\eta_1$ indicates higher spatial correlation among sensors. We set the number of stages $K'$ to $10$ since it was shown to be sufficient in the above figure. We can see that the proposed Deep CoSR-AA schemes outperform both the original CoSR-AA scheme and the ADMM-CSNet \cite{yang_sun_li_xu_2020}, regardless of the correlation length. This shows that the proposed deep unfolding approach is able to achieve
better performance with significantly lower computational complexity. 

\section{Conclusion}\label{sec.conclusion}

In this work, we proposed a compressed sensor caching framework that takes advantage of the high spatio-temporal correlation of sensor observations to improve caching efficiency and reduce communication cost. Here, each cache is only able to store a limited number of observations sampled from sensors within its coverage. To improve the sparse data recovery at the local caches, we first proposed the CoSR-AA algorithm to enable collaborative data recovery among caches. The proposed method is based on the consensus ADMM algorithm but further exploits the sparsity of the underlying signal to reduce the number of message exchanges between caches. This is done by aligning the reconstructed observations of only a few anchor nodes. To further reduce computation time, we then proposed a learning-based Deep CoSR-AA algorithm based on the deep unfolding of the ADMM iterations. We obtained a graph neural network structure with an embedded autoencoder that is able to significantly reduce the number of iterations (or stages). Finally, simulation results were provided to demonstrate the effectiveness of the proposed schemes against other baseline algorithms.

\bibliographystyle{IEEEbib}
\bibliography{IEEEabrv,Ref_SensorCachingCS}
\end{document}